\documentclass[a4paper,11pt]{amsart}

\usepackage{amsmath, amsfonts, amssymb, amsthm, amscd}
\usepackage{graphicx}
\usepackage{perpage}
\usepackage{url}
\usepackage{color}
\usepackage[T1]{fontenc}
\usepackage{hyperref}
\usepackage[a4paper,scale={0.72,0.74},marginratio={1:1},footskip=7mm,headsep=10mm]{geometry}

\numberwithin{equation}{section}

\newtheorem{theorem}{Theorem}[section]
\newtheorem{lemma}[theorem]{Lemma}
\newtheorem{proposition}[theorem]{Proposition}

\newtheorem{remark}[theorem]{Remark}

\newtheorem{conjecture}[theorem]{Conjecture}

\newtheorem{example}[theorem]{Example}

\numberwithin{equation}{section}

\newcommand{\gb}{\beta}
          
\newcommand{\gd}{\delta}
\newcommand{\gep}{\varepsilon}

\newcommand{\gk}{\kappa}
\newcommand{\go}{\omega}

\newcommand{\gO}{\Omega}
\newcommand{\gl}{\lambda}

\newcommand{\bbN}{\mathbb{N}}
\newcommand{\bbZ}{\mathbb{Z}}
\newcommand{\bbZd}{\mathbb{Z}^{d}}

\newcommand{\bbP}{\mathbb{P}}

\newcommand{\bbE}{\mathbb{E}}

\newcommand{\cB}{\mathcal{B}}
\newcommand{\cC}{\mathcal{C}}
\newcommand{\cD}{\mathcal{D}}
\newcommand{\cE}{\mathcal{E}}

\newcommand{\cN}{\mathcal{N}}

\newcommand{\cQ}{\mathcal{Q}}

\newcommand{\bE}{{\ensuremath{\mathbf E}} }

\newcommand{\ind}{\mathbf{1}}

\newcommand{\dd}{\mathrm{d}}
\newcommand{\wsaw}{\mathrm{wsaw}}

\newcommand{\R}{\mathbb{R}}

\newcommand{\Z}{\mathbb{Z}}
\newcommand{\N}{\mathbb{N}}

\newcommand{\bbVar}{{\ensuremath{\mathbb{V}\mathrm{ar}}} }

\newcommand{\be}{\begin{equation}}
\newcommand{\ee}{\end{equation}}

\begin{document}

\title{Annealed scaling for a charged polymer 
in dimensions two and higher}

\author{Q.\ Berger}
\address{LPMA,  Universit\'e Pierre et Marie Curie, case 188, 4 Place Jussieu, 
75005 Paris Cedex, France.}
\email{quentin.berger@upmc.fr}

\author{F.\ den Hollander}
\address{Mathematical Institute, Leiden University, P.O.\ Box 9512,
2300 RA Leiden, The Netherlands.}
\email{denholla@math.leidenuniv.nl}
\author{J.\ Poisat}
\address{Universit\'e Paris-Dauphine, PSL Research University, CNRS, 
UMR [7534], CEREMADE, Place du Mar\'echal de Lattre de Tassigny, 
75016 Paris, France.}
\email{poisat@ceremade.dauphine.fr}

\begin{abstract}
This paper considers an undirected polymer chain on $\bbZd$, $d \geq 2$, with i.i.d.\ 
random charges attached to its constituent monomers. Each self-intersection of the 
polymer chain contributes an energy to the interaction Hamiltonian that is equal to the 
product of the charges of the two monomers that meet. The joint probability distribution 
for the polymer chain and the charges is given by the Gibbs distribution associated 
with the interaction Hamiltonian. The object of interest is the \emph{annealed free 
energy} per monomer in the limit as the length $n$ of the polymer chain tends to infinity. 

We show that there is a critical curve in the parameter plane spanned by the charge 
bias and the inverse temperature separating an \emph{extended phase} from a 
\emph{collapsed phase}. We derive the scaling of the critical curve for small and 
for large charge bias and the scaling of the annealed free energy for small inverse 
temperature. We show that in a subset of the collapsed phase the polymer chain is 
\emph{subdiffusive}, namely, on scale $(n/\log n)^{1/(d+2)}$ it moves like a Brownian 
motion conditioned to stay inside a ball with a deterministic radius and a randomly 
shifted center. We expect this scaling to hold throughout the collapsed phase. We 
further expect that in the extended phase the polymer chain scales like a weakly 
self-avoiding walk.

The scaling of the critical curve for small charge bias and the scaling of the annealed 
free energy for small inverse temperature are both anomalous. Proofs are based on 
a detailed analysis for simple random walk of the downward large deviations of the 
self-intersection local time and the upward large deviations of the range. Part of our 
scaling results are rough. We formulate conjectures under which they can be sharpened.
The existence of the free energy remains an open problem, which we are able to settle 
in a subset of the collapsed phase for a subclass of charge distributions.     
\end{abstract}

\thanks{The research in this paper was supported through ERC Advanced Grant 
267356-VARIS}


\date{\today}

\maketitle

\newpage

\tableofcontents

\newpage

\section{Introduction and main results}
\label{s:intro}

In Caravenna, den Hollander, P\'etr\'elis and Poisat~\cite{CdHPP16}, a detailed 
study was carried out of the annealed scaling properties of an undirected 
polymer chain on $\bbZ$ whose monomers carry i.i.d.\ random charges, in 
the limit as the length $n$ of the polymer chain tends to infinity. With the 
help of the \emph{Ray-Knight representation} for the local times of simple 
random walk on $\bbZ$, a \emph{spectral representation} for the annealed 
free energy per monomer was derived. This was used to prove that there 
is a critical curve in the parameter plane spanned by the charge bias and 
the inverse temperature, separating a \emph{ballistic phase} from a 
\emph{subballistic phase}. Various properties of the phase diagram were derived, 
including scaling properties of the critical curve for small and for large charge bias, 
and of the annealed free energy for small inverse temperature and near the 
critical curve. In addition, laws of large numbers, central limit theorems and 
large deviation principles were derived for the empirical speed and the empirical 
charge of the polymer chain in the limit as $n\to\infty$. The phase transition was 
found to be of \emph{first order}, with the limiting speed and charge making a 
jump at the critical curve. The large deviation rate functions were found to have 
\emph{linear pieces}, indicating the occurrence of mixed optimal strategies where 
part of the polymer is subballistic and the remaining part is ballistic.

The Ray-Knight representation is no longer available for $\bbZd$, $d \geq 2$. The 
goal of the present paper is to investigate what can be said with the help of other 
tools. In Section~\ref{ss:model} we define the model, which was originally introduced 
in Kantor and Kardar~\cite{KK91}. In Section~\ref{ss:mainthm} we state our main 
theorems (Theorems~\ref{thm:fe}, \ref{thm:criticalpoint} and \ref{thm:hightemp} below). 
In Section~\ref{ss:disc} we place these theorems in their proper context. In 
Section~\ref{ss:outline} we outline the remainder of the paper and list some open 
questions.

What makes the charged polymer model challenging is that the \emph{interaction is 
both attractive and repulsive}. This places it outside the range of models that have 
been studied with the help of subadditivity techniques (see Ioffe~\cite{I15} for an 
overview), and makes it into a testbed for the development of new approaches. The
\emph{collapse transition} of a charged polymer can be seen as a simplified version 
of the \emph{folding transition} of a protein. Interactions between different parts of 
the protein cause it to fold into different configurations depending on the temperature.  

Throughout the paper we use the notation $\N=\{1,2,\dots\}$ and $\N_0=\N\cup\{0\}$.

\subsection{Model and assumptions}
\label{ss:model}
  
Let $S=(S_i)_{i\in\N_0}$ be simple random walk on $\Z^d$, $d \geq 1$, starting at $S_0=0$.
The path $S$ models the configuration of the polymer chain, i.e., $S_i$ is the location of 
monomer $i$. We use the letters $P$ and $E$ for probability and expectation with respect 
to $S$. 

Let $\omega=(\omega_i)_{i\in\N}$ be i.i.d.\ random variables taking values in $\R$. The 
sequence $\omega$ models the charges along the polymer chain, i.e., $\omega_i$ 
is the charge of monomer $i$ (see Fig.~\ref{fig-charpol}). We use the letters $\bbP$ and 
$\bbE$ for probability and expectation with respect to $\omega$, and assume that
\begin{equation}
\label{eq:omegacond1} 
M(\delta) = \bbE[e^{\delta\omega_1}]<\infty \quad \forall\,\delta \in \R.
\end{equation}
Without loss of generality (see \eqref{gstardef} below) we further assume that 
\begin{equation}
\label{eq:omegacond2} 
\bbE[\omega_1]=0, \qquad \bbE[\omega_1^2]=1.
\end{equation}
To allow for biased charges, we use the parameter $\delta$ to tilt $\bbP$, namely, we write 
$\bbP^\delta$ for the i.i.d.\ law of $\omega$ with marginal
\begin{equation}
\label{eq:Pdeltadef}
\bbP^\delta(\dd \omega_1) = \frac{e^{\delta \omega_1}\,\bbP(\dd \omega_1)}{M(\delta)}.
\end{equation} 
Without loss of generality we may take $\delta \in [0,\infty)$. Note that $\bbE^\delta[\omega_1]
=M'(\delta)/M(\delta)$. 

\begin{example} 
{\rm If the charges are $+1$ with probability $p$ and $-1$ with probability $1-p$ 
for some $p\in (0,1)$, then $\bbP=[\tfrac12(\delta_{-1}+\delta_{+1})]^{\otimes\N}$ and 
$\delta=\tfrac12\log(\frac{p}{1-p})$.} \hfill \qed
\end{example}

Let $\Pi$ denote the set of nearest-neighbour paths on $\Z^d$ starting at $0$. Given $n\in\N$, 
we associate with each $(\omega,S) \in \R^\N \times \Pi$ an energy given by the Hamiltonian 
(see Fig.~\ref{fig-charpol})
\begin{equation}
\label{eq:def.ham}
H_n^\omega(S) = \sum_{1 \leq i < j \leq n} \omega_i \omega_j\, \ind_{\{S_i=S_j\}}.
\end{equation} 
Let $\beta\in (0,\infty)$ denote the inverse temperature. Throughout the sequel the relevant 
space for the pair of parameters $(\delta,\beta)$ is the quadrant
\begin{equation}
\cQ = [0,\infty) \times (0,\infty).
\end{equation}
Given $(\delta,\beta) \in \cQ$, the \emph{annealed polymer measure of length} $n$ is the
Gibbs measure $\bbP_n^{\delta,\beta}$ defined as
\begin{equation}
\label{eq:def.polmeasure}
\frac{\dd\bbP_n^{\delta,\beta}}{\dd(\bbP^\delta \times P)}(\omega,S) 
= \frac{1}{\bbZ_n^{\delta,\beta}}\,e^{-\beta H_n^\omega(S)},
\qquad (\omega,S) \in \R^\N \times \Pi,
\end{equation} 
where
\begin{equation}
\label{eq:def.annpartfunc}
\Z_n^{\delta,\beta} = (\bbE^\delta \times E)\left[e^{-\beta H_n^\omega(S)}\right]
\end{equation}
is the \emph{annealed partition function of length} $n$. The measure $\bbP_n^{\delta,\beta}$ 
is the joint probability distribution for the polymer chain and the charges at charge bias $\delta$
and inverse temperature $\beta$, when the polymer chain has length $n$.

\begin{figure}[htbp]
\vspace{0.5cm}
\begin{center}
\includegraphics[scale = 0.4]{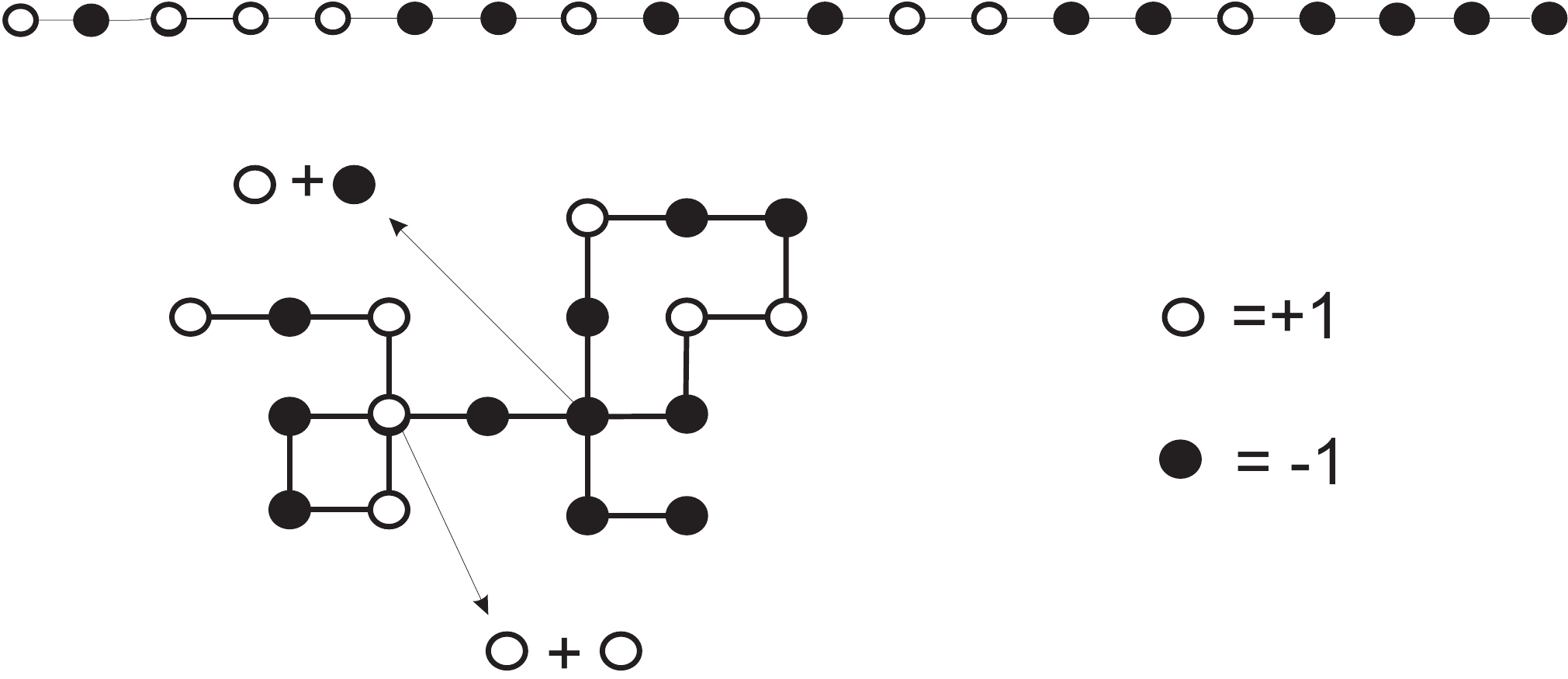}
\end{center}
\caption{\emph{Top:} A polymer chain of length $n=20$ carrying $(\pm 1)$-valued random 
charges. \emph{Bottom:} The charges only interact at self-intersections: in the picture 
monomers $i=4,\,j=8$ meet and repel each other, while monomers $i=10,\,j=18$ meet
and attract each other.}
\label{fig-charpol}
\end{figure}

In what follows, instead of \eqref{eq:def.ham} we will work with the Hamiltonian 
\begin{equation}
\label{eq:def.hamiltonian}
H_n^\omega(S) 
= \sum_{1 \leq i,j \leq n} \omega_i\omega_j\,\ind_{\{S_i=S_j\}}
= \sum_{x\in \Z^d} \left(\sum_{i=1}^n \omega_i\, \ind_{\{S_i=x\}}\right)^2.
\end{equation} 
The sum under the square is the local time of $S$ at site $x$ weighted by the charges that are 
encountered in $\omega$. The change from \eqref{eq:def.ham} to \eqref{eq:def.hamiltonian} 
amounts to replacing $\beta$ by $2\beta$ (to add the terms with $i>j$) and changing the 
charge bias (to add the terms with $i=j$). The latter corresponds to tilting by $\delta\omega_1
+\beta\omega_1^2$ instead of $\delta\omega_1$ in \eqref{eq:Pdeltadef}, which is the same
as shifting $\delta$ by a value that depends on $\delta$ and $\beta$. 

The expression in \eqref{eq:def.annpartfunc} can be rewritten as 
\begin{equation}
\label{Zndbdef}
\bbZ_n^{\gd,\gb} = E \bigg[ \prod_{x\in\bbZd} g_{\gd,\gb}\big(\ell_n(x)\big) \bigg],
\end{equation}
where $\ell_n(x)=\sum_{i=1}^n \ind_{\{S_i=x\}}$ is the local time at site $x$ up to time $n$, and 
\begin{equation}
\label{gdbdef}
g_{\gd,\gb}(\ell) = \bbE^{\gd}\big[\exp(-\gb\gO_{\ell}^2)\big], 
\qquad \gO_{\ell} = \sum_{i=1}^{\ell} \go_i, \qquad \ell \in \N_0\, .
\end{equation}
The \emph{annealed free energy} per monomer is defined by
\begin{equation}
\label{felim}
F(\gd,\gb) = \limsup_{n\to\infty} \frac1n \log \bbZ_n^{\gd,\gb}.
\end{equation}

\begin{remark}
{\rm We expect, but are unable to prove, that the limes superior in \eqref{felim} is a limit. 
A better name for $F$ would therefore be the \emph{pseudo annealed free energy} per 
monomer, but we will not insist on terminology. Convergence appears to be hard to settle, 
due to the competition between attractive and repulsive interactions. Nonetheless, we are 
able to prove convergence for large enough $\beta$ and for charge distributions that are
non-lattice with a bounded density (see Theorem~\ref{thm:conv_free_energy} below).}  
\hfill \qed
\end{remark}

\subsection{Main theorems}
\label{ss:mainthm}

Our first theorem provides relevant upper and lower bounds on $F$. Abbreviate $f(\gd) 
= - \log M(\gd) \in (-\infty,0]$.

\begin{theorem}
\label{thm:fe}
The limes superior in \eqref{felim} takes values in $(-\infty,0]$ and satisfies the inequality 
$F(\gd,\gb) \geq f(\gd)$. \hfill \qed
\end{theorem}

The \emph{excess annealed free energy} per monomer is defined by
\begin{equation}
\label{excessfe}
F^*(\gd,\gb) = F(\gd,\gb) - f(\gd).
\end{equation}
It follows from \eqref{Zndbdef}--\eqref{felim} that
\begin{equation}
\label{Fstardef}
F^*(\gd,\gb) = \limsup_{n\to\infty} \frac1n \log \bbZ_n^{*,\gd,\gb}
\end{equation}
with
\begin{equation}
\label{Zndbstar}
\bbZ_n^{*,\gd,\gb} = E\bigg[ \prod_{x\in\bbZd} g_{\gd,\gb}^*\big(\ell_n(x)\big) \bigg],
\end{equation}
where
\begin{equation}
\label{gstardef}
g_{\gd,\gb}^*(\ell) = \bbE\big[ \exp\big(\gd\gO_{\ell}-\gb\gO_{\ell}^2\big)\big],
\qquad \ell \in \N_0.
\end{equation}
(This expression shows why the assumption in \eqref{eq:omegacond2} respresents no 
loss of generality.) We may think of $g_{\gd,\gb}^*(\ell)$ as a single-site partition function 
for a site that is visited $\ell$ times.

\begin{example} 
{\rm If the distribution of the charges is standard normal, then
\begin{equation}
\label{Gauss}
g_{\gd,\gb}^*(\ell) = \sqrt{\frac{1}{1+2\gb\ell}}\,\exp\left[\frac{\gd^2\ell}{2(1+2\gb\ell)}\right],
\qquad \ell \in \N_0.
\end{equation}
Note that $-\log g_{\gd,\gb}^*$ can be decomposed as $-\log g_{\gd,\gb}^* 
= -\log g^{*,\mathrm{att}}_{\gd,\gb} - \log g^{*,\mathrm{rep}}_{\gd,\gb}$ with 
\begin{equation}
- \log g^{*,\mathrm{att}}_{\gd,\gb}(\ell) = \frac{1}{2}\,\log(1+2\beta\ell), \qquad 
- \log g^{*,\mathrm{rep}}_{\gd,\gb}(\ell) = - \frac{\delta^2 \ell}{2(1+2\beta\ell)}.  
\end{equation}
The former is an attractive interaction (positive concave function), the latter is a repulsive 
interaction (negative convex function).} \hfill \qed
\end{example}

\medskip
Because $F^*(\gd,\gb) \geq 0$, it is natural to define two phases:
\begin{equation}
\label{phasesdef}
\begin{aligned}
\cC &= \{(\gd,\gb)\in\cQ\colon\,F^*(\gd,\gb)=0\},\\
\cE &= \{(\gd,\gb)\in\cQ\colon\,F^*(\gd,\gb)>0\}.
\end{aligned}
\end{equation}
For reasons that will become clear later, we refer to these as the \emph{collapsed phase}, 
respectively, the \emph{extended phase}. For every $\delta \in [0,\infty)$, $\gb \mapsto 
F^*(\gd,\gb)$ is finite, non-negative, non-increasing and convex. Hence there is a critical 
threshold $\beta_c(\delta) \in [0,\infty]$ such that $\cC$ is the region on and above the 
curve and $\cE$ is the region below the curve (see Fig.~\ref{fig-critcurve}).

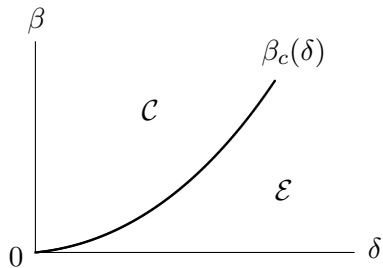
\begin{figure}[htbp]
\begin{center}
\setlength{\unitlength}{0.35cm}
\begin{picture}(12,12)(0,-2)
\put(0,0){\line(12,0){12}}
\put(0,0){\line(0,8){8}}
{\thicklines
\qbezier(0,0)(5,0.5)(9,6.5)
}
\put(-1,-.5){$0$}
\put(12.5,-0.2){$\delta$}
\put(-0.3,8.5){$\beta$}
\put(8.5,7.3){$\beta_c(\delta)$}
\put(9,2){$\cE$}
\put(4,5){$\cC$}
\end{picture}
\end{center}
\vspace{0cm}
\caption{\small Qualitative plot of the critical curve $\delta \mapsto \beta_c(\delta)$ 
where the excess free energy $F^*(\delta,\beta)$ changes from being zero
($\cC$) to being strictly positive ($\cE$). The critical curve is part of $\cC$.}
\label{fig-critcurve}
\end{figure}

Our second theorem describes the qualitative properties of the critical curve, provides scaling 
bounds for small charge bias, and identifies the asymptotics for large charge bias. Let 
\begin{equation}
\label{Qndef}
Q_n = \sum_{x\in\bbZd} \ell_n(x)^2
\end{equation} 
denote the \emph{self-intersection local time} at time $n$. A standard computation 
gives (see e.g.\ Spitzer~\cite[Section 7]{S76}), as $n\to\infty,$
\begin{equation}
\label{eq:asympQ}
E[Q_n] = \sum_{1 \leq i,j \leq n} P(S_i=S_j) \sim
\begin{cases}
\gl_2 n \log n, &d=2,\\
\gl_d n, &d \geq 3,
\end{cases}
\end{equation}
with
\begin{equation}
\label{gld}
\gl_2 = 2/\pi, \qquad \gl_d = 2G_d-1, \quad d \geq 3, 
\end{equation}
where $G_d=\sum_{n\in\N_0} P(S_n=0)$ is the Green function at the origin of simple random 
walk on $\Z^d$. A similar computation yields (see Chen~\cite[Sections 5.4--5.5]{C10}) 
\begin{equation}
\label{eq:asympQalt}
\mathrm{Var}(Q_n) = E[Q_n^2] - E[Q_n]^2 \sim 
\begin{cases}
C_2 n^2, &d=2,\\
C_3 n\log n, &d=3,\\
C_d n, &d \geq 4,
\end{cases} 
\end{equation}
with $C_d$, $d \geq 2$, computable constants. In particular, $Q_n$ satisfies the weak law 
of large numbers.

Abbreviate $m_k=\bbE[\go_1^k]$, $k\in\N$, and recall that $m_1=0$, $m_2=1$ by 
\eqref{eq:omegacond2}.

\begin{theorem}
\label{thm:criticalpoint}
{\rm (i)} $\gd \mapsto \beta_c(\gd)$ is continuous, strictly increasing and convex on 
$[0,\infty)$, with $\beta_c(0)=0$.\\
{\rm (ii)} As $\gd\downarrow 0$,
\begin{equation}
\label{betacexp}
\gb_c(\gd) = \tfrac12 \gd^2 - \tfrac13 m_3\gd^3 - \gep_\gd
\end{equation}
with
\begin{equation} 
[\underline{\kappa}+o(1)]\,\gd^4
\le
\gep_\gd \le [1+o(1)] \begin{cases}
\kappa_2 \gd^4 \log(1/\gd), &d=2,\\
\kappa_d \gd^4, &d\geq 3,
\end{cases}
\end{equation}
where
\begin{equation}
\label{kappadid}
\underline{\kappa} = \tfrac{1}{12} m_4 - \tfrac13 m_3^2,
\qquad
\gk_d =
\begin{cases}
\tfrac14 \gl_2 , &d=2\\
\tfrac14 (\gl_d-1) + \underline{\kappa}, &d \geq 3. 
\end{cases}
\end{equation}
{\rm (iii)} As $\delta\to\infty$, 
\begin{equation} 
\label{eq:betacasympinf}
\beta_c(\delta) \sim \frac{\delta}{T}
\end{equation}
with 
\begin{equation} 
\label{eq:Tlat}
T = \sup\big\{t > 0\colon\,\bbP(\omega_1 \in t\,\Z) = 1\big\}
\end{equation}
(with the convention $\sup\emptyset = 0$). Either $T>0$ (`lattice case') or $T=0$ 
(`non-lattice case'). If $T=0$ and $\omega_1$ has a bounded density (with respect 
to the Lebesgue measure), then
\begin{equation}
\beta_c(\delta) \sim \frac{\delta^2}{4 \log \delta}.
\end{equation}
\hfill \qed
\end{theorem}

Our third theorem offers scaling bounds on the free energy for small inverse
temperature and fixed charge bias.

\begin{theorem}
\label{thm:hightemp}
For any $\gd \in (0,\infty)$, as $\gb \downarrow 0$,
\begin{equation}
\label{eq:Fscal}
-\big[m(\gd)^2 + v(\gd) + o(1)\big]\,\gb \ge F(\gd,\gb) \ge [1+o(1)] \left\{
\begin{array}{ll}
- \gl_2 m(\gd)^2\,\gb\log(1/\gb),  &d=2,\\[0.2cm]
- \big[\gl_d m(\gd)^2 + v(\gd)\big]\,\gb, &d\geq 3,
\end{array}
\right.
\end{equation}
where $m(\gd) = \bbE^\delta[\omega_1]$ and $v(\gd) = \bbVar^\gd[\omega_1]$. 
\hfill \qed
\end{theorem}

Our fourth and last main theorem settles existence of the free energy for large enough
inverse temperature for a subclass of charge distributions.  

\begin{theorem}
\label{thm:conv_free_energy}
Suppose that the charge distribution is non-lattice $(T=0)$ and has a bounded density.
Then there exists a curve $\gd \mapsto \gb_0(\gd)$ such that, for all $\gb \ge \gb_0(\gd)$,
\begin{enumerate}
\item 
the sequence $\{\log g^*_{\gd,\gb}(\ell)\}_{\ell\in\N}$ is super-additive,
\item 
the limes superior in \eqref{felim} is a limit, and equals $-f(\gd)$,
\item 
the limes superior in \eqref{Fstardef} is a limit, and equals $0$.
\end{enumerate}
Moreover, $\gb_0(\gd) \ge \gb_c(\gd)$ and $\gb_0(\gd) \sim \gb_c(\gd)$ as $\gd\to\infty$.
\hfill\qed
\end{theorem}

\subsection{Discussion and two conjectures}
\label{ss:disc}

We discuss the theorems stated in Section~\ref{ss:mainthm} and place them in their 
proper context. 

\medskip\noindent
{\bf 1.}
Theorem~\ref{thm:fe} shows that the annealed excess free energy $(\gd,\gb) \mapsto 
F^*(\gd,\gb)$ is nonnegative on $\cQ$ and satisfies a lower bound that signals the 
presence of two phases. 

\medskip\noindent
{\bf 2.}
Theorem~\ref{thm:criticalpoint}(i) shows that there is a phase transition at a non-trivial 
critical curve $\gd \mapsto \gb_c(\gd)$ in $\cQ$, separating a collapsed phase $\cC$ 
(on and above the curve) from an extended phase $\cE$ (below the curve). If the charge 
distribution is \emph{symmetric}, then 
\begin{equation}
\label{curveub}
\gb_c(\gd) \leq \tfrac12\gd^2 \qquad \forall\,\gd \in [0,\infty).
\end{equation} 
Indeed, using \eqref{gstardef} we may estimate
\begin{equation}
\begin{aligned}
g_{\gd,\tfrac12\gd^2}^*(\ell) 
&= \bbE\left[ \exp\big(\gd\gO_{\ell}-\tfrac12\gd^2\gO_{\ell}^2\big)\right]
= \bbE\left[ \sum_{k\in\N_0} \frac{1}{k!} (\gd\gO_{\ell})^k\,
\exp\big(-\tfrac12\gd^2\gO_{\ell}^2\big)\right]\\
&= \bbE\left[ \sum_{k\in\N_0} \frac{1}{(2k)!} (\gd\gO_{\ell})^{2k}\,
\exp\big(-\tfrac12\gd^2\gO_{\ell}^2\big)\right]\\
&\leq \bbE\left[ \sum_{k\in\N_0} \frac{1}{k!} (\tfrac12\gd^2\gO_{\ell}^2)^k\,
\exp\big(-\tfrac12\gd^2\gO_{\ell}^2\big)\right]
= \bbE[1] = 1 \qquad \forall\,\ell\in\N_0,
\end{aligned}
\end{equation}
where we use that $(2k)! \geq 2^k\,k!$, $k \in \N_0$. Via \eqref{Fstardef}--\eqref{Zndbstar}
this implies that $\Z_n^{*,\gd,\tfrac12\gd^2} \leq 1$ for all $n\in\N$ and hence $F^*(\delta,
\tfrac12\gd^2) = 0$, which via \eqref{phasesdef} yields \eqref{curveub} (see Fig.~\ref{fig-critcurve}).

\medskip\noindent
{\bf 3.} 
The lower and upper bounds in Theorem~\ref{thm:criticalpoint}(ii) differ by a multiplicative 
factor when $d \ge 3$ and by a logarithmic factor when $d=2$. We expect that the upper 
bound gives the right asymptotic behaviour:

\begin{conjecture} 
\label{conj:sharp_beta_c}
As $\gd \downarrow 0$,
\begin{equation}
\gep_\gd \sim \begin{cases}
\kappa_2 \gd^4 \log(1/\gd), &d=2,\\
\kappa_d \gd^4, &d\geq 3.
\end{cases}
\end{equation}
\hfill\qed
\end{conjecture}

\noindent
In Appendix \ref{app:range} we state a conjecture about trimmed local times that 
would imply Conjecture \ref{conj:sharp_beta_c}. Theorem~\ref{thm:criticalpoint}(ii) 
identifies three terms in the upper bound of $\gb_c(\delta)$ for small $\gd$, of which 
the last is \emph{anomalous} for $d=2$. The proof is based on an analysis of the 
\emph{downward large deviations} of the self-intersection local time $Q_n$ in 
\eqref{Qndef} under the law $P$ of simple random walk in the limit as $n\to\infty$. 
A sharp result was found in Caravenna, den Hollander, P\'etr\'elis and Poisat~\cite{CdHPP16} 
for $d=1$, with two terms in the expansion of which the last is anomalous (namely, 
order $\delta^{8/3}$). For the standard normal distribution $m_3=0$ and $m_4=3$, 
and so $\kappa_d = \frac14\lambda_d$ for $d \geq 2$ in \eqref{kappadid}. 

\medskip\noindent
{\bf 4.} 
Note that $\kappa_d \ge \underline{\gk}>0$ for $d \geq 3$ when $m_3=0$, but not 
necessarily when $m_3 \neq 0$. Indeed, if the distribution of the charges puts weight 
$\frac{1}{3N^2}$, $1-\tfrac{1}{2N^2}$, $\tfrac{1}{6N^2}$ on the values $-N$, $0$, $2N$, 
respectively, for some $N\in\N$, then $m_1=0$, $m_2=1$, $m_3=N$, $m_4=3N^2$, 
in which case $-\tfrac{1}{3}m_3^2+\tfrac{1}{12}m_4 = -\tfrac{1}{12}N^2$. This gives 
$\kappa_d<0$ for $N$ large enough and $\underline{\gk} < 0  \le \gk_d$ for $N$ 
small enough.

\medskip\noindent
{\bf 5.}
Theorem~\ref{thm:criticalpoint}(iii) identifies the asymptotics of $\gb_c(\delta)$ for 
large $\gd$, which is the same as for $d=1$. The scaling depends on whether the 
charge distribution is lattice or non-lattice.

\medskip\noindent
{\bf 6.} 
In analogy with what we saw in Theorem~\ref{thm:criticalpoint}(ii), the bounds in 
Theorem~\ref{thm:hightemp} do not match, but we expect the following:

\begin{conjecture}
\label{conj:asym.smallbeta}
For any $\gd \in (0,\infty)$, as $\gb \downarrow 0$,
\begin{equation}
F(\gd,\gb) \sim \left\{
\begin{array}{ll}
- \gl_2 m(\gd)^2\,\gb\log(1/\gb),  &d=2,\\[0.2cm]
- \big[\gl_d m(\gd)^2 + v(\gd)\big]\,\gb, &d\geq 3,
\end{array}
\right.
\end{equation}
\hfill \qed
\end{conjecture}

\noindent
This identifies the scaling behaviour of the free energy for small inverse temperature (i.e., in the limit 
of weak interaction). The scaling is anomalous for $d=2$, as it was in \cite{CdHPP16} for $d=1$ 
(namely, order $\beta^{2/3}$).

\medskip\noindent
{\bf 7.}
Theorem~\ref{thm:conv_free_energy} settles the existence of the free energy in a subset of the 
collapsed phase for a subclass of charge distributions. The limit is expected to exist always.

\medskip\noindent
{\bf 8.}
As shown in den Hollander~\cite[Chapter 8]{dH09}, for every $d \geq 1$ and every 
$(\gd,\gb) \in \mathrm{int}(\cC)$,
\begin{equation}
\label{fannlimit}
\lim_{n\to\infty} \frac{(\alpha_n)^2}{n} \, \log \bbZ_n^{*,\gd,\gb}  = -\chi_d,
 \end{equation} 
with $\alpha_n = (n/\log n)^{1/(d+2)}$ and with $\chi_d \in (0,\infty)$ a constant that is explicitly 
computable. The idea behind \eqref{fannlimit} is that the empirical charge makes a large 
deviation under the law $\bbP^\delta$ so that it becomes zero. The price for this large deviation 
is 
\begin{equation}
e^{-nH(\bbP^0 \,|\, \bbP^\delta) + o(n)}, \qquad n\to\infty,
\end{equation}
where $H(\bbP^0 \,|\, \bbP^\delta)$ denotes the specific relative entropy of $\bbP^0=\bbP$ 
with respect to $\bbP^\delta$. Since the latter equals $\log M(\delta)=-f(\delta)$, this accounts 
for the term that is subtracted in the excess free energy. Conditional on the empirical charge 
being zero, the attraction between charged monomers with the same sign \emph{wins} from 
the repulsion between charged monomers with opposite sign, making the polymer chain 
contract to a \emph{subdiffusive} scale $\alpha_n$. This accounts for the correction term in 
the free energy. It is shown in \cite{dH09} that, under the law $\bbP^\gd$,
\begin{equation}
\left(\frac{1}{\alpha_n}\,S_{\lfloor nt\rfloor}\right)_{0 \leq t \leq 1} 
\Longrightarrow (U_t)_{0 \leq t \leq 1}, \qquad n\to\infty,
\end{equation}
where $\Longrightarrow$ denotes convergence in distribution and $(U_t)_{t \geq 0}$ is 
a Brownian motion on $\R^d$ conditioned not to leave a ball with a deterministic radius 
and a randomly shifted center (see Fig.~\ref{fig-Browconf}). Compactification is a key 
step in the sketch of the proof provided in den Hollander~\cite[Chapter 8]{dH09}, 
which requires super-additivity of $\{\log g^*_{\gd,\gb}(\ell)\}_{\ell\in\N}$. From 
Theorem~\ref{thm:conv_free_energy}(1) we know that this property holds at least for 
$\beta$ large enough.

\begin{figure}[htbp]
\vspace{-1cm}
\begin{center}
\includegraphics[scale = 0.25]{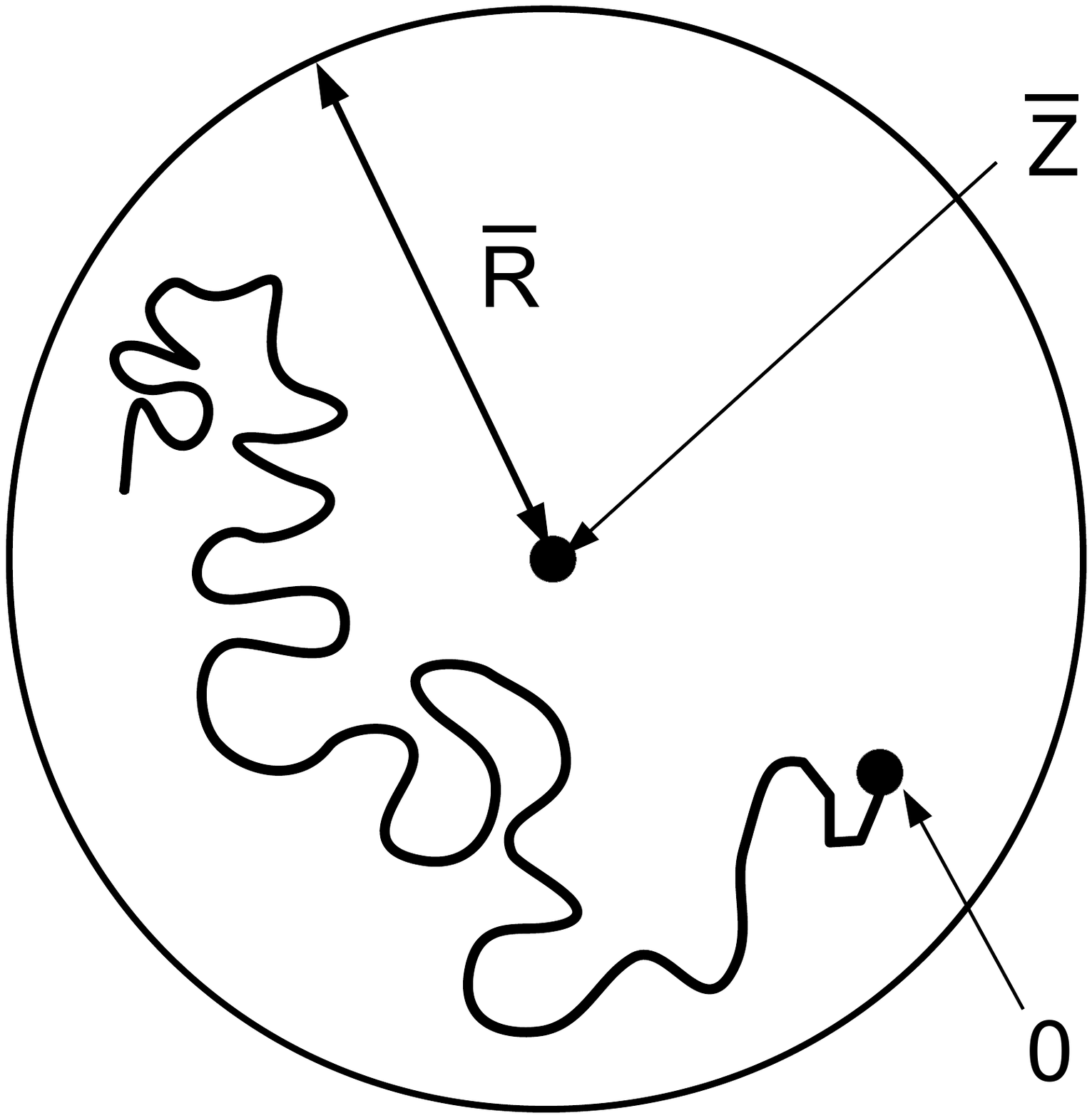}
\end{center}
\vspace{-1.5cm}
\caption{A Brownian motion starting at $0$ conditioned to stay inside the ball with radius $\bar{R}$ 
and center $\bar{Z}$. Formulas for $\bar{R}$ and the distribution of $\bar{Z}$, concentrated on the 
ball of radius $\bar{R}$ centered at $0$, are given in \cite[Chapter~8]{dH09}.}
\label{fig-Browconf}
\vspace{-.3cm}
\end{figure}

\medskip\noindent
{\bf 9.}
It is natural to expect that for every $(\gd,\gb) \in \cE$ the polymer behaves like \emph{weakly 
self-avoiding walk}. Once the empirical charge is strictly positive, the repulsion should
win from the attraction, and the polymer should scale as if all the charges were strictly 
positive, with a change of time scale only.  

\medskip\noindent
{\bf 10.}
Brydges, van der Hofstad and K\"onig~\cite{BvdHK07} derive a formula for the 
joint density of the local times of a continuous-time Markov chain on a finite graph, using 
tools from finite-dimensional complex calculus. This representation, which is the analogue 
of the Ray-Knight representation for the local times of one-dimensional simple random walk, 
involves a large determinant and therefore appears to be intractable for the analysis of 
the annealed charged polymer.

\subsection{Outline and open questions} \label{ss:outline}

The remainder of this paper is organised as follows. In Section~\ref{s:WSAW} we study 
the downward large deviations of the self-intersection local time $Q_n$ defined in \eqref{Qndef} 
under the law $P$ of simple random walk. We derive the qualitative properties of the rate 
function, which amounts to controlling the partition function (and free energy) of weakly self-avoiding walk 
with the help of cutting arguments. In Section~\ref{s:feexist} we prove Theorem~\ref{thm:fe}. 
In Section~\ref{s:Gstar} we prove Theorem~\ref{thm:criticalpoint}. The proof of part (i) 
requires a detailed analysis of the function $\ell \mapsto g_{\gd,\gb}^*(\ell)$ defined in 
\eqref{gstardef}. The proof of part (ii) is based on estimates of the function $\ell \mapsto 
g_{\gd,\gb}^*(\ell)$ for small values of $\gd$. The proof of part (iii) carries over from 
\cite{CdHPP16}. In Section~\ref{s:scalingfreeenergy} we use the results in Section~\ref{s:WSAW} 
to prove Theorem~\ref{thm:hightemp}, and in Section~\ref{s:conv_free_energy} we prove 
Theorem~\ref{thm:conv_free_energy}. In Appendix~\ref{app:bridge} we collect some 
estimates on simple random walk constrained to be a bridge, which are needed along 
the way. In Appendix~\ref{app:wsaw} we state a conjecture on weakly self-avoiding walk 
that complement the results in Section~\ref{s:WSAW}. In Appendix~\ref{app:range} we 
discuss a rough estimate on the probability of an upward large deviation for the range of 
simple random walk, trimmed when the local times exceed a given threshold. This estimate 
appears to be the key to Conjectures~\ref{conj:sharp_beta_c} and \ref{conj:asym.smallbeta}.

\medskip\noindent
Here are some \emph{open questions}:
\begin{itemize}
\item[(1)]
Is the limes superior in \eqref{felim} always a limit? For $d=1$ the answer was found to be yes.
\item[(2)]
Is $(\gd,\gb) \mapsto F^*(\gd,\gb)$ analytic throughout the extended phase 
$\cE$? For $d=1$ the answer was found to be yes.
\item[(3)]
How does $F^*(\gd,\gb)$ behave as $\beta\uparrow\gb_c(\gd)$? Is the phase 
transition first order, as for $d=1$, or higher order?
\item[(4)]
Is the excess free energy monotone in the dimension, i.e., $F^{*\,(d+1)}(\gd,\gb) 
\geq F^{*\,(d)}(\gd,\gb)$ for all $(\gd,\gb) \in \cQ$ and $d \geq 1$?
\item[(5)]
What is the nature of the expansion of $\beta_c(\gd)$ for $\gd \downarrow 0$,
of which \eqref{betacexp} gives the first three terms? Is it anomalous with a logarithmic 
correction to the term of order $\delta^{2d}$ for any $d \geq 3$?
\end{itemize}

\section{Weakly self-avoiding walk}
\label{s:WSAW}

In Section~\ref{ss:SILT} we look at the free energy $f^{\wsaw}$ of the weakly self-avoiding walk, 
identify its scaling in the limit of weak interaction (Proposition~\ref{pr:wsaw} below). 
In Section~\ref{ss:lem1} we look at the rate function for the 
downward large deviations of the self-intersection local time $Q_n$ as $n\to\infty$ 
(Proposition~\ref{pr:largedevsuml2} below). In Section~\ref{ss:lem2} we use this rate function 
to prove the scaling of $f^{\wsaw}$.

\begin{remark}
\label{rem:bridge}
{\rm Let $\cB_n$ be the set of $n$-step bridges 
\begin{equation}
\label{bridgedef}
\cB_n = \left\{S \in \Pi\colon\,0 = S^{(1)}_0 < S^{(1)}_i < S^{(1)}_n 
\,\,\forall\, 0 < i < n\right\},
\end{equation}
where $S^{(1)}$ stands for the first coordinate of simple random walk $S$. At several points 
in the paper we will use that there exists a $C \in (0,\infty)$ such that
\begin{equation}
\label{eq:asymp.bridge}
\lim_{n\to\infty} n\,P(S \in \cB_n) = C,
\end{equation}
a property we will prove in Appendix~\ref{app:bridge0}.} \hfill \qed
\end{remark}

\subsection{Self-intersection local time}
\label{ss:SILT}

Recall the definition of the self-intersection local time $Q_n=\sum_{x\in\Z^d} \ell_n(x)^2$ 
in \eqref{Qndef}. For $u\geq 0$, let
\begin{equation}
Z_n^{\wsaw}(u) = E\big[e^{-u Q_n}\big], \qquad u \in [0,\infty),
\end{equation}
be the partition function of weakly self-avoiding walk. This quantity is submultiplicative
because $Q_{n+m}\geq Q_{n}+Q_m$, $m,n\in\N$. Hence (minus) the free energy of 
the weakly self-avoiding walk
\begin{equation}
\label{eq:fwsawdef}
f^{\wsaw}(u) = - \lim_{n\to\infty} \frac{1}{n} \log Z_n^{\wsaw}(u), \qquad u \in [0,\infty),
\end{equation} 
exists. The following lemma identifies the scaling behaviour of $f^{\wsaw}(u)$ for 
$u \downarrow 0$.

\begin{proposition}
\label{pr:wsaw}
As $u\downarrow0$
\begin{equation}
\label{eq:fwsaw}
f^{\wsaw}(u) \sim 
\begin{cases}
\gl_1 u^{1/3}, &d=1,\\
\gl_2 u \log(1/u), &d=2,\\
\gl_d u, &d\geq 3,
\end{cases} 
\end{equation}
where $\gl_d$ is given in \eqref{gld}. \hfill \qed
\end{proposition}

\noindent
Proposition~\ref{pr:wsaw} extends the downward moderate deviation result for $Q_n$ 
derived by Chen~\cite[Theorem 8.3.2]{C10}. For more background on large 
deviation theory, see den Hollander~\cite{dH00}. We comment further on this result in 
Appendix~\ref{app:wsaw}, where we discuss the rate of convergence to $f^{\wsaw}(u)$ 
and the higher order terms in the asymptotic expansion of $f^{\wsaw}(u)$ as $u\downarrow 0$.

\subsection{Downward large deviations of the self-intersection local time}
\label{ss:lem1}

In Section~\ref{ss:lem2} we will show that Proposition~\ref{pr:wsaw} is a consequence of 
the following lemma describing the downward large deviation behaviour of $Q_n$ (see 
Fig.~\ref{fig-ratefunc}).

\begin{proposition}
\label{pr:largedevsuml2}
The limit
\begin{equation}
\label{Itdef}
I(t) = \lim_{n\to\infty} \left[-\frac{1}{n}\log P(Q_n \leq tn)\right], \qquad t \in [1,\infty),
\end{equation}
exists. Moreover, $t \mapsto I(t)$ is finite, non-negative, non-increasing and convex
on $[1,\infty)$, and satisfies
\begin{equation}
\label{Iasymp}
\begin{aligned}
d=2\colon \quad I(t) > 0, \quad t \geq 1,
\qquad
d \geq 3\colon \quad I(t)
\left\{\begin{array}{ll}
> 0, &1 \leq t \leq \gl_d,\\
= 0, &t \geq \gl_d.
\end{array}
\right. 
\end{aligned}
\end{equation}
Furthermore,
\begin{equation}
\label{It2scal}
d=2\colon \quad \lim_{t\to\infty} \frac{-\log I(t)}{t} = \frac{1}{\gl_2}.
\end{equation}
\hfill \qed
\end{proposition}

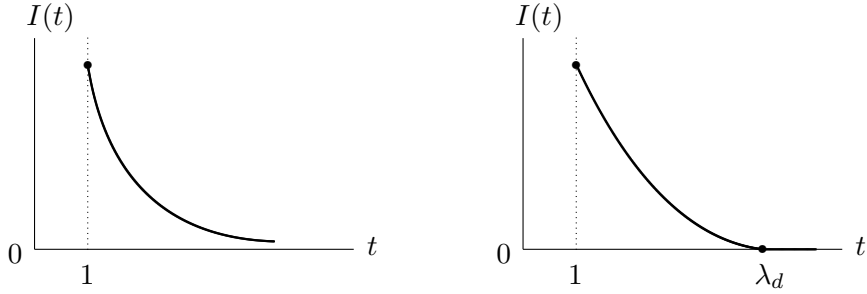
\begin{figure}[htbp]
\begin{center}
\setlength{\unitlength}{0.35cm}
\begin{picture}(12,12)(4,-2)
\put(0,0){\line(12,0){12}}
\put(0,0){\line(0,8){8}}
{\thicklines
\qbezier(2,7)(3,0.5)(9,.3)
}
\qbezier[40](2,0)(2,4)(2,8)
\put(-1,-.5){$0$}
\put(12.5,-0.2){$t$}
\put(-0.3,8.5){$I(t)$}
\put(1.7,-1.3){$1$}
\put(2,7){\circle*{.35}}
\end{picture}
\begin{picture}(12,12)(-2,-2)
\put(0,0){\line(12,0){12}}
\put(0,0){\line(0,8){8}}
{\thicklines
\qbezier(2,7)(5,0.5)(9,0)
\qbezier(9,0)(10,0)(11,0)
}
\qbezier[40](2,0)(2,4)(2,8)
\put(-1,-.5){$0$}
\put(12.5,-0.2){$t$}
\put(-0.3,8.5){$I(t)$}
\put(1.7,-1.3){$1$}
\put(8.7,-1.3){$\gl_d$}
\put(2,7){\circle*{.35}}
\put(9,0){\circle*{.35}}
\end{picture}
\end{center}
\vspace{0cm}
\caption{\small Qualitative plots of $t \mapsto I(t)$ for $d=2$ and $d \geq 3$.}
\label{fig-ratefunc}
\end{figure}

\begin{proof}
The proof comes in 5 Steps. Steps 1--2 use bridges and superadditivity, Steps 3--5 
use cutting arguments.

\medskip\noindent
{\bf 1.\ Existence, finiteness and monotonicity of $I$.} 
Recall \eqref{bridgedef}. Let $\cB_n$ be short for $\{S \in \cB_n\}$. Define
\begin{equation}
u(n) = P(Q_n \leq tn,\, \cB_n),\quad n\in\bbN.
\end{equation}
The sequence $(\log u(n))_{n\in\N}$ is superadditive. Therefore $\lim_{n\to\infty} [-\frac{1}{n}
\log u(n)] = \bar{I}(t) \in [0,\infty]$ exists. Clearly,
\begin{equation}
\label{PQnest1}
\limsup_{n\to\infty} \left[-\frac{1}{n} \log P(Q_n \leq tn)\right] \leq \bar{I}(t).
\end{equation}
The reverse inequality follows from a standard unfolding procedure applied to bridges that 
decreases $Q_n$. Indeed, using the bound introduced in Hammersley and Welsh~\cite{HW62}, 
we get
\begin{equation}
|\{Q_n \leq tn\}| \leq e^{\pi \sqrt{\frac{n}{3}}(1+o(1))} |\{Q_n \leq tn\}\cap \cB_n|,
\end{equation}
from which it follows that 
\begin{equation}
\label{PQnest2}
\liminf_{n\to\infty} \left[-\frac{1}{n} \log P(Q_n \leq tn)\right] \geq \bar{I}(t).
\end{equation}
Combining \eqref{PQnest1} and \eqref{PQnest2}, we get \eqref{Itdef} with $I=\bar{I}$. Finally, it is 
obvious that $t \mapsto I(t)$ is non-increasing on $[1,\infty)$. Since $\{Q_n =n\} = \{(S_i)_{i=0}^n 
\text{ is self-avoiding}\}$, we have $I(1) = \log \mu_c(\bbZd) < \infty$, with $\mu_c(\bbZd)$ the 
connective constant of $\bbZd$.

\medskip\noindent
{\bf 2.\ Convexity of $I$.}
Every $2n$-step walk $S_{[0,2n]}= (S_i)_{0\leq i\leq 2n}$ can be decomposed into two $n$-step
walks: $S_{[0,n]}=(S_i)_{0\leq i \leq n}$ and $\bar S_{[0,n]} = (S_{n+i} - S_n)_{0\leq i\leq n}$. Fix 
$a,b>0$. Restricting both parts to be a bridge, we get
\begin{equation}
\begin{aligned}
P(Q_{2n} \leq (a+b) n,\, \cB_{2n}) 
&\geq P\Big( Q_n \leq a n,\,\bar Q_n \leq b n,\,S \in \cB_n\,,\bar S \in \cB_n \Big) \\
&= P\big( Q_n \leq a n,\, S \in \cB_n \big) \, P\big( Q_n \leq b n,\, S \in  \cB_n \big),
\end{aligned}
\end{equation}
where $\bar Q_n = \sum_{1\leq i,j \leq n} \ind_{\{\bar S_i = \bar S_j\}}$. Taking 
the logarithm, diving by $2n$ and letting $n \to \infty$, we get
\begin{equation}
I\left(\tfrac12(a+b)\right) \leq \tfrac12 [I(a) + I(b)].
\end{equation} 

\medskip\noindent
{\bf 3. Two regimes of $I$ for $d\geq 3$.} 
Clearly, $I(t)=0$ for $t\geq \gl_d$. To prove that $I(t)>0$ for $1\leq t <\gl_d$, we 
cut $[0,n]$ into sub-intervals of length $1/\eta$, where $\eta>0$ is small and 
$\eta n$ is integer. Note that
\begin{equation}
\label{eq:cutQn}
Q_n \geq \sum_{1 \leq k \leq \eta n} Q^{(k)}, \qquad 
Q^{(k)} = \sum_{\frac{k-1}{\eta}+1 \leq i,j \leq \frac{k}{\eta}} \ind_{\{S_i = S_j\}}.
\end{equation}
Fix $\gep >0$ small. Then, by \eqref{eq:asympQ}, there exists an $\eta_\gep$ such 
that $E[Q^{(1)}] \geq \frac{1}{\eta}(\gl_d - \gep^2)$ for $0< \eta \leq \eta_\gep$. Moreover,
by the Markov property of simple random walk, the $Q^{(k)}$'s are independent. 
Therefore we may estimate, for $\gamma>0$,
\begin{equation}
\label{Chernov}
\begin{aligned}
&P\big(Q_n \leq (\gl_d - \gep)n\big) 
\leq P\left(-\gamma \sum_{1\leq k\leq \eta n} Q^{(k)} \geq - \gamma (\gl_d-\gep)n \right)\\ 
&\qquad \leq e^{\gamma (\gl_d - \gep)n} E\big[e^{-\gamma Q^{(1)}}\big]^{\eta n}  \leq 
e^{\gamma (\gl_d - \gep)n} \Big(1 - \gamma E[Q^{(1)}] 
+ \tfrac12\gamma^2 E[(Q^{(1)})^2]  \Big)^{\eta n}\\
&\qquad \leq e^{\gamma (\gl_d - \gep)n}\,e^{\big(-\gamma E[Q^{(1)}] 
+ \tfrac12\gamma^2 E[(Q^{(1)})^2]\big)\,\eta n}
\leq  e^{- n\gamma \big(\gep - \tfrac12\eta\gamma E[(Q^{(1)})^2]\big)}.
\end{aligned}
\end{equation}
Because $Q^{(1)} \leq 1/\eta^2$ (and hence $E[(Q^{(1)})^2] \leq 1/\eta^4$), it suffices to 
choose $\gamma$ small enough to get from \eqref{Itdef} that $I(\gl_d - \gep)>0$. Since
$\gep>0$ is arbitrary, this proves the claim.

\medskip\noindent
{\bf 4.\ Positivity and asymptotics of $I$ for $d=2$.}
To obtain a lower bound on the probability $P(Q_n \leq tn)$ we use a specific strategy, 
explained informally in Fig.~\ref{fig-strategy}. Let $\gep>0$ and 
\begin{equation}
\label{mchoice}
m = \lfloor e^{\frac{t}{(1+\gep)\gl_2}}\rfloor\geq 2.
\end{equation}

\begin{figure}[htbp]
\begin{center}
\setlength{\unitlength}{0.4cm}
\begin{picture}(30,8)(0,0)
\put(0,0){\includegraphics[width=12cm,height=3cm]{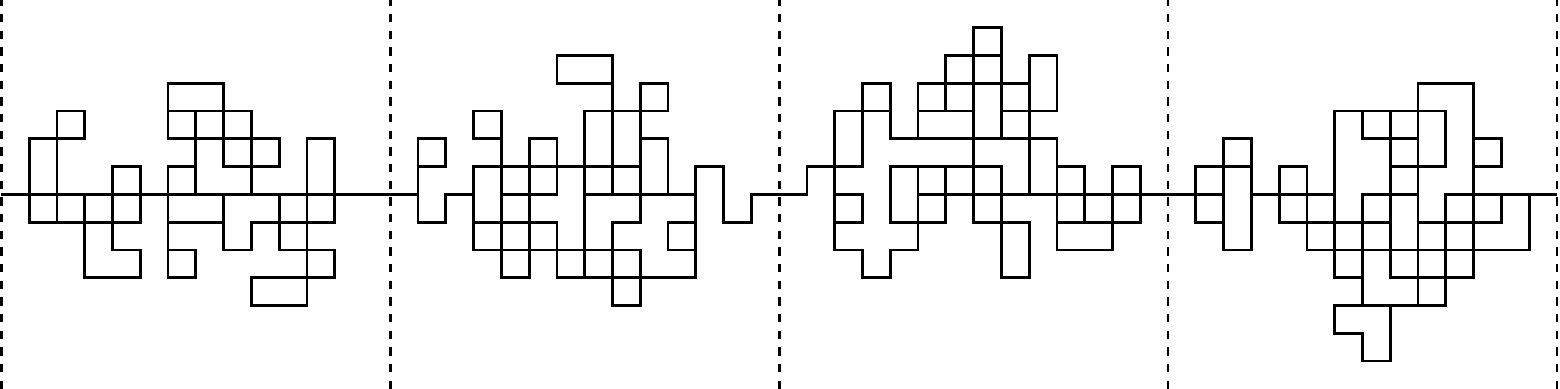} }
\put(1.2,0.3){\tiny $m_t\approx e^{t/\gl_2}$ steps}
\put(.5,-.7){\tiny $Q_m \lesssim \gl_2 m\log m$}
\end{picture}
\end{center}
\vspace{0cm}
\caption{\small \emph{Informal description of the specific strategy to obtain $Q_n\leq tn$}: 
Confine $(S_i)_{i=0}^n$ to $n/m$ consecutive strips, each containing $m \approx e^{t/\gl_2}$ 
steps. On each strip impose the walk to be a bridge. By \eqref{eq:asympQ}, each strip 
contributes $\lesssim \gl_2 m \log m$ to the self-intersection local time, and hence
$Q_n \lesssim  \frac{n}{m}(\gl_2 m\log m) \approx t n$. The cost per bridge is $\approx
1/m$. Consequently, the cost of the consecutive strip strategy is $(1/m)^{n/m} \approx 
\exp(-n m^{-1} \log m)$. Hence $I(t)\lesssim  m^{-1} \log m = c\,te^{-t/\gl_2}$.}
\label{fig-strategy}
\end{figure}

\noindent
For $n\in \N$, write $n = pm+q$, where $p = p(n) \in \N_0$  and $0 < q = q(n) \leq m$. 
For $k\in\N$, define the events
\begin{equation}
\begin{aligned} \label{eq:blockbridge}
U_k &= \Big\{ S^{(1)}_{(k-1)m} \leq S^{(1)}_{i} \leq S^{(1)}_{km - 1}\,\,
\forall\, (k-1)m < i < km,\, S^{(1)}_{km} =S^{(1)}_{km-1} + 1 \Big\},\\
V_k &= \{Q^{(k)} \leq (1+\gep)\gl_2 m \log m\}, 
\end{aligned}
\end{equation}
with $Q^{(k)}$ as in \eqref{eq:cutQn} with $1/\eta=m$, 
and
\begin{equation}
W = \left[\bigcap_{k=1}^p U_k \cap V_k \right] \bigcap 
\left[\bigcap_{j=1}^q \left\{S^{(1)}_{pm+j} = S^{(1)}_{pm}+j\right\}\right].
\end{equation}
Note that, on the event $W$,
\begin{equation}
Q_n = \sum_{k=1}^p Q^{(k)} \leq (1+\gep)\gl_2 \,p\, m \log m \leq tn.
\end{equation}
Hence
\begin{equation}
P(Q_n \leq tn) \geq P(Q_n \leq tn,\,W)
\geq \Big[ \frac{1}{4}\,P\big(Q_m \leq (1+ \gep)\gl_2 m \log m,\, S \in \cB_m\big) \Big]^p 
\Big(\frac{1}{4} \Big)^q.
\end{equation}
We therefore obtain
\begin{equation}
\frac{1}{n}\log P(Q_n \leq tn) \geq \frac{1 - \frac{q}{n}}{m}
\Big[\log P\big(Q_m \leq (1 + \gep) \gl_2 m \log m,\, S \in \cB_m\big) - \log 4  \Big] 
- \frac{q}{n} \log 4
\end{equation}
and, by taking the limit $n\to \infty$, we get
\begin{equation}
\label{liminfPQn}
\liminf_{n\to\infty} \frac{1}{n}\log P(Q_n \leq tn) \geq \frac{1}{m} 
\Big[\log P\big(Q_m \leq (1+\gep)\gl_2 m \log m,\, S \in \cB_m\big) - \log 4  \Big].
\end{equation}
In Appendix \ref{app:bridge1} we prove that
\begin{equation}
\label{liminfPQnalt}
P\big(Q_m \leq (1+\gep)\gl_2 m \log m,\, S \in \cB_m\big) \sim P(S \in \cB_m),
\qquad m \to \infty. 
\end{equation}
Therefore, by \eqref{eq:asymp.bridge}, the right-hand side of \eqref{liminfPQn} scales like 
$-\log m/m$ as $m\to\infty$. Combining \eqref{Itdef}, \eqref{mchoice} and
 \eqref{liminfPQn}--\eqref{liminfPQnalt}, we arrive at
\begin{equation}
\label{It2upper}
I(t) \leq  \frac{t}{(1+\gep)\gl_2}\,e^{-\frac{t}{(1+\gep)\gl_2}}[1+o(1)], \qquad t\to\infty.
\end{equation}
This proves that $\liminf_{t\to\infty} -\log I(t)/t \geq 1/(1+\gep)\gl_2$. Let $\gep \downarrow 0$ 
to get the lower half of \eqref{It2scal}.

\medskip\noindent
{\bf 5.} 
To obtain an upper bound on the probability $P(Q_n \leq tn)$ we use the same 
type of strategy. Let $\gep >0$, choose $m$ large enough so that $E[Q^{(1)}] 
\geq (1-\gep) \gl_2  m \log m$, and use that there exists a constant $c$ such 
that $E[Q_n^2] \leq c (n\log n)^2$. Cut $[0,n]$ into sub-intervals of length $m$, 
similarly as in \eqref{eq:cutQn} with $m$ instead of $1/\eta$ (assume that $n/m$ 
is integer). Estimate
\begin{align}
\label{Pestlink}
P(Q_n \leq tn) &\leq  P\Big(\sum_{1\leq i \leq n/m} Q^{(i)} \leq t n \Big) 
\leq e^{\gamma t n} E\big[e^{-\gamma Q^{(1)}} \big]^{n/m} \nonumber\\
&\leq  e^{\gamma tn} e^{\frac{n}{m}\big(- \gamma E[Q^{(1)}] 
+ \tfrac12 \gamma^2 E[(Q^{(1)})^{2}] \big)} 
\leq e^{\gamma tn} e^{\frac{n}{m}\big( -\gamma (1-\gep) \gl_2 m \log m 
+ c\tfrac12 \gamma^2 m^2 (\log m)^2\big)}.
\end{align}
Choose $m = \lfloor e^{\frac{1+\gep}{1-\gep}\frac{t}{\gl_2}}\rfloor$, which diverges 
as $t\to\infty$. Then \eqref{Pestlink} becomes
\begin{equation}
P(Q_n \leq tn) \leq e^{- n \gamma \big(-t\gep + c\tfrac12\gamma m (\log m)^2\big)}.
\end{equation}
Optimizing over $\gamma$, i.e., choosing $\gamma = t\gep/c\,m (\log m)^2$, 
we get
\begin{equation}
P(Q_n \leq tn) \leq \exp\Big( - c(\gep) e^{-\frac{1+\gep}{1-\gep}\frac{t}{\gl_2}} n\Big)
\end{equation}
for some constant $c(\gep)>0$, and so we arrive at
\begin{equation}
I(t)\geq c(\gep)\,e^{-\frac{1+\gep}{1-\gep}\frac{t}{\gl_2}}, \qquad t \to \infty.
\end{equation}
This proves that $\limsup_{t\to\infty} -\log I(t)/t \leq (1+\gep)/(1-\gep)\gl_2$.
Let $\gep \downarrow 0$ to get the upper half of \eqref{It2scal}, which completes 
the proof of Proposition~\ref{pr:largedevsuml2}.
\end{proof}

\begin{remark}\rm
\label{rem:adaptdownard}
We may adapt the argument in Step 4 to obtain a result that will be needed in
\eqref{eq:dowarddevConstrained} below, namely, a lower bound on the probability
\begin{equation}
v_n(t) = P \left(Q_n \leq tn,\, \max_{x\in\bbZ^2} \ell_n(x) 
\leq c_1e^{c_2t}\right)
\end{equation}
with $c_1>0$, $c_2=\big( 2\gl_2(1+\tfrac14\gep) \big)^{-1}$ and $\gep>0$ small. 
This lower bound reads
\begin{equation}
\label{lbrelevant}
\liminf_{n\to\infty} \frac{1}{n}\log v_n(t) \geq - \frac{t}{(1+\gep)\gl_2}\, 
e^{-\frac{t}{(1+\gep)\gl_2}}[1+o(1)], \qquad t\to\infty.
\end{equation}
Indeed, the strategy above is still valid, and \eqref{liminfPQn} becomes
\begin{equation}
\begin{aligned}
&\liminf_{n\to\infty} \frac{1}{n}\log v_n(t)\\
&\qquad \geq \frac{1}{m} 
\Big[\log P\Big(Q_m \leq (1+\gep)\gl_2 m \log m,\, \max_{x\in\bbZ^2} \ell_m(x)
\leq  c_1m^{c_3},\, S \in \cB_m \Big) - \log 4  \Big]
\end{aligned}
\end{equation}
with $m$ as in \eqref{mchoice} and $c_3=\tfrac12(1+\gep)/(1+\tfrac14\gep)$. Since the local 
times are typically of order $\log m$, the constraint on the maximum of the local times is 
harmless in the limit as $m\to\infty$ and can be removed. After that we obtain \eqref{lbrelevant}
following the argument in \eqref{liminfPQn}--\eqref{liminfPQnalt}. To check that the constraint 
can be removed, estimate
\begin{equation}
\begin{aligned}
&P\Big(\max_{x\in\bbZ^2} \ell_m(x) >  c_1m^{c_3}\Big) 
\leq m P\big(\ell_m(0) >  c_1m^{c_3}\big)\\
&\qquad \leq m \left(1-\frac{c_4}{\log m}\right)^{c_1 m^{c_3}} 
\leq m\,e^{-c_1c_4 m^{c_3} \log m},
\end{aligned}
\end{equation}
which is $o(1/m)$. \hfill \qed
\end{remark}

\subsection{Scaling of the free energy of weakly self-avoiding walk}
\label{ss:lem2}

In this section we prove Proposition~\ref{pr:wsaw}.

\begin{proof}
From Proposition~\ref{pr:largedevsuml2} and Varadhan's lemma we obtain 
\begin{equation}
\label{eq:Varadhan}
- f^{\wsaw} (u) = \sup_{t \in [1,\infty)} [-tu -I(t)].
\end{equation}

\medskip\noindent
{\bf Upper bound:}
For $d\geq 3$, choose $t=\gl_d$ and use that $I(\gl_d)=0$, to obtain $-f^{\wsaw} (u) 
\geq -\gl_du$ for all $u$, which is the upper half of \eqref{eq:fwsaw}. 

For $d=2$, by \eqref{It2scal}, for any $\gep>0$ we have $I(t)\leq e^{-(1-\gep) t/\gl_2}$ 
for $t$ large enough. Choose $t=(1-\gep)^{-1}\gl_2\log(1/u)$ to obtain $-f^{\wsaw} (u) \geq - 
(1-\gep)^{-1}\gl_2u\log(1/u)-u $, so that 
\begin{equation}
\limsup_{u\downarrow 0 } \frac{f^{\wsaw}(u)}{u\log(1/u) }
\leq (1-\gep)^{-1} \gl_2.
\end{equation}
Let $\gep\downarrow 0$ to get the upper half of \eqref{eq:fwsaw}.

\medskip\noindent
{\bf Lower bound:}
For $d\geq 3$, write
\begin{equation}
- f^{\wsaw}(u) = \sup_{1\leq t\leq \gl_d} [-tu-I(t)] 
= -\gl_d u + \sup_{1 \leq t \leq \gl_d} [(\gl_d - t)u-I(t)].
\end{equation}
Fix $\gep>0$ small. Then $I(\gl_d - \gep) >0$. By convexity, $I(t) \geq \frac{\gl_d - t}{\gep} 
I(\gl_d-\gep)$ for all $1 \leq t \leq \gl_d - \gep$. Therefore
\begin{equation}
- f^{\wsaw}(u) \leq -\gl_d u + \sup_{1 \leq t \leq \gl_d - \gep} 
\left[(\gl_d - t)u - \tfrac{\gl_d - t}{\gep} I(\gl_d - \gep)\right] 
\vee \sup_{\gl_d - \gep< t\leq \gl_d} [(\gl_d - t)u - I(t)].
\end{equation}
For $u \leq I(\gl_d - \gep)/\gep$ the first supremum is non-positive and the second supremum 
is at most $\gep u$. This implies that $f^{\wsaw}(u) \geq (\gl_d-\gep) u$ for $u$ small 
enough (namely, $u \leq I(\gl_d - \gep)/\gep$). Let $\gep \downarrow 0$ to get the lower
half of \eqref{eq:fwsaw}.

For $d=2$, by \eqref{It2scal}, for any $\gep>0$ we have $I(t)\geq e^{-(1+\gep) t/\gl_2}$ for 
$t$ large enough. We have
\begin{equation}
\begin{aligned}
-f^{\wsaw} (u) &\leq \sup_{1 \leq t \leq t_0 } [-tu - I(t)] \
\vee\, \sup_{t \geq t_0} [-tu -I(t)]\\
&\leq \sup_{1 \leq t \leq t_0 } [- I(t)] \vee\, \sup_{t\geq t_0} \left[-tu -  e^{-(1+\gep) t/\gl_2}\right]\\
&= - (1+\gep)^{-1} \gl_2 u \log(1/u) + O(u),
\end{aligned}
\end{equation}
where the first supremum is simply a constant and the last supremum is attained at 
$t = - (1+\gep)^{-1}\gl_2 \log((1+\gep)^{-1}\gl_2 u)$, which is larger than $t_0$ for 
$u$ small enough. Let $\gep\downarrow 0$ to get the lower half of \eqref{eq:fwsaw}.
\end{proof}

\section{Bounds on the annealed free energy}
\label{s:feexist}

In this section we prove Theorem~\ref{thm:fe}. It is obvious from \eqref{Zndbdef}--\eqref{felim} 
that $F(\gd,\gb) \leq 0$. The lower bound $F(\gd,\gb) \geq -f(\gd)$ is derived by forcing 
simple random walk to stay inside a ball of radius $\alpha_n=(n/\log n)^{1/(d+2)}$ centered 
at the origin. Indeed, let $\cE_n=\{S_i \in B(0,\alpha_n)\,\,\forall\,0 \leq i \leq n\}$. Then, by 
\eqref{Zndbstar},
\begin{equation}
\bbZ_n^{*,\gd,\gb} \geq E\bigg[ \ind_{\cE_n} 
\prod_{x\in\bbZd} g_{\gd,\gb}^*\big(\ell_n(x)\big) \bigg].
\end{equation}
As shown in Lemma~\ref{lem:estimG}(2) below, we have $g_{\delta,\beta}^*(\ell) \asymp 
1/\sqrt{\ell}$ as $\ell\to\infty$. Hence there exists a $c>0$ such that 
\begin{equation}
\bbZ_n^{*,\gd,\gb} \geq E\bigg[ \ind_{\cE_n} \exp\bigg(-c\sum_{x\in\bbZd} 
\log \ell_n(x)\bigg) \bigg].
\end{equation}
Since $\sum_{x\in\bbZd} \ell_n(x) = n$, Jensen's inequality gives
\begin{equation}
\bbZ_n^{*,\gd,\gb} \geq E\Big[ \ind_{\cE_n} \exp\Big(-c R_n \log \frac{n}{R_n}\Big)\Big] 
\end{equation}
with $R_n = |\{x\in\bbZd\colon\,\ell_n(x)>0\}|$ the range up to time $n$. On the event 
$\cE_n$, we have $R_n = O(\alpha_n^d) = o(n)$, $n\to\infty$. Hence there exists 
a $c'>0$ such that
\begin{equation}
\bbZ_n^{*,\gd,\gb} \geq P(\cE_n) \exp\Big(-c' \alpha_n^d \log n\Big).
\end{equation}
But $P(\cE_n) = \exp(-[1+o(1)]\mu_d n/\alpha_n^2)$ with $\mu_d$ the principal Dirichlet 
eigenvalue of the Laplacian on the ball in $\R^d$ of unit radius centered at the origin. 
Hence
\begin{equation}
F^*(\gd,\gb) = \limsup_{n\to\infty} \frac{1}{n} \log \bbZ_n^{*,\gd,\gb}  \geq 0,
\end{equation}
which proves the claim (recall \eqref{excessfe}).

\section{Critical curve}
\label{s:Gstar}

In Section~\ref{ss:proofcritgen} we prove Theorem~\ref{thm:criticalpoint}(i). In 
Section~\ref{ss:Gstarsmalldelta} we derive lower and upper bounds on 
$g^*_{\gd,\gb}$ for small $\gd,\gb$ (Lemma~\ref{lem:estimG} below). In 
Sections~\ref{ss:proofscalcritsmall-lb} and \ref{ss:proofscalcritsmall-ub} we 
combine these bounds with Proposition~\ref{pr:largedevsuml2} and a detailed 
study of the cost of ``rough local-time profiles'' of simple random walk, in order
to derive lower and upper bounds, respectively, on the critical curve for small 
charge bias (Lemma~\ref{lem:lbcritcurve} below; see also Lemma~\ref{lem:rough}). The latter 
bounds imply Theorem~\ref{thm:criticalpoint}(ii). In 
Section~\ref{ss:proofscalcritlarge} we prove Theorem~\ref{thm:criticalpoint}(iii), 
which carries over from \cite{CdHPP16}.

\subsection{General properties of the critical curve}
\label{ss:proofcritgen}

\begin{proof}
The proof is standard. Fix $\gd\in [0,\infty)$. Clearly, $\gb \to F^*(\gd,\gb)$ is non-increasing 
and convex on $(0,\infty)$, and hence is continuous on $(0,\infty)$. Moreover, from Jensen's 
inequality we get $F^*(\gd,0)=-f(\gd) \geq F^*(\gd,\gb) \geq -f(\gd) - \gb$, so $\gb \to F^*(\gd,\gb)$ 
is actually continuous on $[0,\infty)$.

By Theorem~\ref{thm:fe}, we know that $F^*(\gd,\gb) \geq 0$. Since $\gb \mapsto 
F^*(\gd,\gb)$ is non-increasing and continuous, there exists a $\gb_c(\gd)=\sup\{\gb\in 
(0,\infty)\colon\,F^*(\gd,\gb)>0\}$ such that $F^*(\gd,\gb)>0$ when $0<\gb<\gb_c(\gd)$ 
and $F^*(\gd,\gb)=0$ when $\gb \geq \gb_c(\gd)$. Since $(\gd,\gb) \mapsto F^*(\gd,
\gb)$ is convex on $\cQ$, the level set $\{(\gd,\gb)\in \cQ \colon\,F^*(\gd,\gb) \leq 0\}$ 
is convex, and it follows that $\gd \mapsto \gb_c(\gd)$ (which coincides with the 
boundary of this level set) is also convex. 

First, fix $\gd\in [0,\infty)$. We prove that $\gb_c(\gd)<\infty$ by showing that, for 
$\gb$ large enough, $g_{\gd,\gb}^*(\ell)\leq 1$ for all $\ell\in\N$, which implies that 
$F^*(\gd,\gb)=0$. Indeed, by choosing $\gep>0$ small enough and cutting the integral 
in \eqref{gstardef} according to whether $|\gO_\ell| \leq \gep$ or $|\gO_\ell|> \gep$, 
we get
\begin{equation}
\label{eq:boundgstar}
g_{\gd,\gb}^*(\ell) \leq e^{\frac{\gd^2}{4\gb}}\ 
\bbP(|\gO_\ell | \leq \gep) + e^{-\gb \gep^2 + \gd \gep}.
\end{equation}
By the Local Limit Theorem, we know that $\lim_{\ell\to\infty} \bbP(|\gO_\ell | \leq \gep)=0$, 
so that $\sup_{\ell \in \N} \bbP(|\gO_\ell | \leq \gep) <1$ provided $\gep$ is small enough. 
The claim follows by choosing $\gb$ large enough in \eqref{eq:boundgstar}. (This argument 
corrects a mistake in \cite[Section 3.1]{CdHPP16}.)

Next, fix $\gd\in (0,\infty)$. Then $F^*(\gd,0)=-f(\gd)>0$, and so $\gb_c(\gd)>0$ by 
continuity. Finally, since $F^*(0,\gb) = 0$ for $\gb \in (0,\infty)$, we get $\gb_c(0)=0$.

The convexity of $\gd \mapsto \gb_c(\gd)$ and the fact that $\gb_c(\gd)>0$ for 
$\gd \in (0,\infty)$ imply that $\gd \mapsto \gb_c(\gd)$ is strictly increasing. The 
continuity of $\gd \mapsto \gb_c(\gd)$ follows from convexity and finiteness.
\end{proof}

\subsection{Estimates on the single-site partition function}
\label{ss:Gstarsmalldelta}
 
In this section we derive estimates on $g_{\gd,\gb}^*$ for $\gd$ small. 
 
\begin{lemma}
\label{lem:estimG}
Let 
\begin{equation}
\label{gbgddef}
\gb(\gd) =\tfrac12\gd^2-\tfrac13m_3\gd^3 -\gep_{\gd}, \qquad \gep_{\gd}=o(\gd^3),
\quad \gd \downarrow 0.
\end{equation}
Then for all $\eta \in (0,1)$ there exist $\gd_0>0$ and $a>0$ such that the following hold:\\
{\rm (1)} If $0 < \gd \leq \gd_0$ and $\gd^2 \ell \leq a$, then
\begin{gather}
\label{eq:lowboundG1}
g_{\gd,\gb(\gd)}^* (\ell) \geq 1+ (\gep_{\gd}+k_1 \gd^4 ) \ell - \tfrac14(1+\eta)\gd^4 \ell^2,\\
g_{\gd,\gb(\gd)}^* (\ell) \leq 1+ (\gep_{\gd} +k_1 \gd^4) \ell - \tfrac14(1-\eta)\gd^4 \ell^2,
\label{eq:upboundG1}
\end{gather}
where 
\begin{equation}
\label{eq:k1def}
k_1 = \tfrac13 m_3^2 - \tfrac{1}{12} m_4 + \tfrac14.
\end{equation}
{\rm (2)} If $0 < \gd \leq \gd_0$ and $\gd^2 \ell \geq a$, then there exists a $c_0>0$ 
such that
\begin{equation}
\label{eq:boundG2}
1 \geq \min\left(1,\frac{c_0}{\sqrt{1+\gd^2 \ell}} \right) 
\geq g_{\gd,\gb(\gd)}^* (\ell) \geq \frac{1}{c_0\sqrt{1+\gd^2 \ell}}.
\end{equation}
\hfill \qed
\end{lemma}

\begin{proof}
Below, all error terms refer to $\gd \downarrow 0$. Fix $\gb=\gb(\gd)$. Write 
$g_{\gd,\gb}^* (\ell) =\bbE[e^X]$ with $X=-\gb\, \gO_{\ell}^2 + \gd\, \gO_{\ell}$. 
The proof is based on asymptotics of moments of $X$ for small $\gd,\gb$.
Recall that $\bbE[\omega_1]=0$, to compute
\begin{equation}
\begin{aligned}
&\bbE[\gO_{\ell}] = 0,  \quad \bbE[\gO_{\ell}^2] = m_2 \ell, 
\quad \bbE[\gO_{\ell}^3] = m_3 \ell,\\
&\bbE[\gO_\ell^4] = 3 m_2^2 \ell (\ell-1) + m_4 \ell, \quad
\bbE[\gO_\ell^5] = 10 m_2 m_3 \ell (\ell-1) + m_5 \ell,\\ 
&\bbE[\gO_\ell^6] = 15 m_2^3 \ell (\ell-1)(\ell-2)  
+ (15m_2 m_4 +10m_3^2) \ell(\ell-1) + m_6 \ell.
\end{aligned}
\end{equation}
If $\gb\asymp \gd^2$, then (recall that $m_2=1$) 
\begin{equation}
\begin{aligned}
&\bbE[X] = -\gb \ell,\\
&\bbE[X^2] = [\gd^2 -2\gb\gd m_3 + \gb^2 k_2] \ell + 3 \gb^2 \ell^2,\\
&\bbE[X^3] = [\gd^3 m_3 - 3 \gb \gd^2 k_2 +o(\gd^4)] \ell 
+ [-9 \gb\gd^2+o(\gd^4)]\ell^2 -15 \gb^3 \ell^3,\\
&\bbE[X^4] =  [k_2\gd^4 +o(\gd^4)] \ell + [3\gd^4 +o(\gd^4)] \ell^2 
+ [90 \gb^2\gd^2 + o(\gd^6)] \ell^3 + [\tfrac{1}{24} \gb^4 + o(\gd^8)] \ell^4,\\
&\bbE[X^5] = o(\gd^4) \ell + o(\gd^4) \ell^2 + c\gd^6 [1+o(1)] \ell^3 
+ c'\gd^8 [1+o(1)] \ell^4 + c''\gd^{10}[1+o(1)] \ell^5,
\end{aligned}
\end{equation}
where $k_2 = m_4-3$, so that $\bbE[\gO_\ell^4] = 3 \ell^2 + k_2 \ell$.
Therefore
\begin{equation}
\begin{aligned}
\bbE[X]+&\tfrac12 \bbE[X^2]+\tfrac16 \bbE[X^3]+\tfrac{1}{24} \bbE[X^4] \\
&= \big[ -\gb m_2 +\tfrac{\gd^2}{2} m_2 - \gb \gd m_3 +\tfrac{\gb^2}{2} k_2 
+ \tfrac16 \gd^3 m_3 - \tfrac{1}{2} \gb\gd^2k_2 +\tfrac{1}{24} \gd^4 k_2 + o(\gd^4) \big] \ell\\
& \qquad + \big[ \tfrac{3}{2} m_2^2 \gb^2 - \tfrac{3}{2} m_2^2 \gb\gd^2 
+ \tfrac{1}{8} m_2^2 \gd^4 + o(\gd^4) \big] \ell^2 
+ O(\gd^6 \ell^3) +O(\gd^8 \ell^4). 
\end{aligned}
\end{equation}
Inserting $m_2=1$ and $\gb = \gb(\gd)$, we get
\begin{equation}
\begin{aligned}
1+\bbE[X]+&\tfrac12 \bbE[X^2]+\tfrac16 \bbE[X^3]+\tfrac{1}{24} \bbE[X^4]\\
&= 1+ \Big[ \gep_{\gd} +  \Big(\tfrac13 m_3^2 - \tfrac{1}{12} k_2 \Big) \gd^4 \Big] \ell 
- \tfrac14\gd^4[1+o(1)] \ell^2 + O(\gd^6 \ell^3) +O(\gd^8 \ell^4),
\end{aligned}
\end{equation}
where we use that $o(\gd^4)\ell = o(\gd^4) \ell^2$. We also get $\bbE[X^k] 
= \sum_{j=\lceil k/2\rceil k}^k O(\gd^{2j} \ell^j)$ for $k \geq 5$.

\medskip\noindent
(1) To obtain the lower bound in \eqref{eq:lowboundG1}, use that $e^x \geq 1+ \sum_{j=2}^5 
\tfrac{1}{j!} x^j$, $x \in \R$, to get
\begin{equation}
\begin{aligned}
g_{\gd,\gb}^*(\ell) &= \bbE[e^X]\\ 
&\geq 1+ (\gep_{\gd}+k_1 \gd^4) \ell 
- \tfrac{1}{4} \gd^4[1+o(1)] \ell^2 + O(\gd^6 \ell^3) +O(\gd^8 \ell^4) + O(\gd^{10} \ell^5),
\end{aligned}
\end{equation}
from which the claim follows for $\gd^2\ell$ small enough. To obtain the upper bound in 
\eqref{eq:upboundG1}, use that $e^x \leq 1+ \sum_{j=2}^6 \tfrac{1}{j!} x^j +\frac{1}{7!} 
x^7 \ind_{\{x\geq 0\}}$, $x \in \R$. Also use that $X=-\gb \gO_{\ell}^2 +\gd \gO_{\ell} \leq 
\gd^2/4\gb \leq 1$, because $\gb\geq \tfrac14\gd^2$ for $\gd$ small enough, which implies
that $\bbE[X^7 \ind_{\{X\geq 0\}}] \leq \bbE[X^6]$. Hence
\begin{equation}
\begin{aligned}
g_{\gd,\gb}^*(\ell) &= \bbE[e^X]\\
&\leq 1+ (\gep_{\gd}+k_1 \gd^4) \ell - \tfrac{1}{4} \gd^4[1+o(1)] \ell^2 
+ O(\gd^6 \ell^3) + O(\gd^8 \ell^4) + O(\gd^{10} \ell^5) + O(\gd^{12} \ell^6),
\end{aligned}
\end{equation}
from which the claim follows for $\gd^2\ell$ small enough.

\medskip\noindent
(2) We fix $b>0$ large, and treat the cases $a<\gd^2 \ell <b$ and $\gd^2\ell \geq b$ separately. 
Since in both cases $\ell\to\infty$ as $\gd\downarrow 0$, we have that $\gO_{\ell}/\sqrt{\ell}$ is
close in distribution to $Z = \cN(0,1)$.

\medskip\noindent
$\bullet$ If $a<\gd^2 \ell <b$, then, uniformly for $a<\gd^2 \ell <b$,
\begin{equation}
\label{eq:approxGgaussian}
g_{\gd,\gb}^*(\ell) = [1+o(1)]\, \bbE\big[e^{-(\gb \ell) Z^2 + \gd\sqrt{\ell} Z}\big] 
= [1+o(1)]\, \bbE\big[e^{-[1+O(\gd)]\,\frac12 (\gd^2 \ell) Z^2 + \gd\sqrt{\ell} Z}\big].
\end{equation}
The function
\begin{equation}
t \mapsto h(t)=\bbE[e^{-\frac12 t^2 Z^2 +t Z}] 
= \frac{1}{\sqrt{1+t^2}}\, e^{\frac12\, \frac{t^2}{1+t^2}}
\end{equation} 
is strictly decreasing with $h(0)=1$. Therefore, for $\gd$ small enough, we find that 
\begin{equation}
\frac{1}{2\sqrt{1+\gd^2 \ell}} \leq g_{\gd,\gb}^*(\ell)  \leq \frac{2}{\sqrt{1+\gd^2 \ell}} 
\end{equation}
(note that $e^{1/2}<2$). Using that $\gd^2 \ell \geq a$ and $h(a)<1$, we obtain
$g_{\gd,\gb}^*\leq 1$.

\medskip\noindent
$\bullet$ If $\gd^2 \ell \geq b$, then we argue as follows. Let $\Phi$ be the standard 
normal cumulative distribution function. Write $Z_{\ell} = \Omega_{\ell}/\sqrt{\ell}$, and 
estimate
\begin{equation}
g_{\gd,\gb}^*(\ell) \geq \bbP(X\geq 0) = \bbP (\gO_{\ell} \in[ 0, \gd/\gb])
= \bbP\left(Z_{\ell} \in \big[0,2/\gd \sqrt{\ell}\big] \right) 
\geq \frac{1}{4\sqrt{\gd^2 \ell}},
\end{equation}
where the last inequality follows from the Berry-Esseen inequality (Feller~\cite[Theorem XVI.5.1]{F71})
\begin{equation}
\label{eq:berryesseen}
\sup_{x\in\R} \left| \bbP\left(Z_{\ell}\leq x\right) - \Phi(x) \right| \leq A / \sqrt{\ell},
\end{equation}
in combination with the bound $|\Phi(0) - \Phi(2/\gd \sqrt{\ell})| \geq 1/3\gd\sqrt{\ell}$,
valid for $\gd^2 \ell \geq b$ with $b$ large enough, and $(1/3\gd \sqrt{\ell}) - (2A/\sqrt{\ell}) 
\geq 1/4\sqrt{\gd^2 \ell}$, valid for $\gd$ small enough. 

To get an upper bound on $g_{\gd,\gb}^*(\ell)$, abbreviate $v=\gd\sqrt{\ell}$ and 
$X = -\frac12 v^2 Z_{\ell}^2 + v Z_{\ell}$, and estimate
\begin{eqnarray}
g_{\gd,\gb}^*(\ell) &\leq& 
\sum_{k=2}^{\log v } e^{-k} P\big(- k \geq X \geq -(k+1) \big)  
+ e^{-\log v}  P(X \leq -\log v) \nonumber\\
&\leq& \sum_{k=2}^{\log v} e^{-k} P\big(v Z_{\ell} \in [1-\sqrt{1+2k},1+\sqrt{1+2k}]\big) 
+ \frac1v \nonumber\\
&\leq& \sum_{k=2}^{\log v} e^{-k} \sqrt{k} \, \frac3v  
+ \frac1v  = C \, \frac1v = \frac{C}{\sqrt{\gd^2 \ell}}, 
\end{eqnarray}
where in the last inequality we again use the Berry-Esseen inequality in \eqref{eq:berryesseen}, 
this time with $|x|, |y|\leq\frac{2}{v} \sqrt{k}$: if $v= \gd\sqrt{\ell} \geq b$ with $b$ large enough, 
then $|\Phi(x) - \Phi(y)|\leq \tfrac12 |x-y| \leq \frac2v \sqrt{k}$, while if $\gd$ is small enough, 
then $2A/\sqrt{\ell} \leq \frac1v \leq \frac1v \sqrt{k}$.
\end{proof}

\subsection{Lower bound on the critical curve for small charge bias}
\label{ss:proofscalcritsmall-lb}

In this section we prove the lower bound in Theorem~\ref{thm:criticalpoint}(ii). Substitute 
\eqref{eq:lowboundG1} into \eqref{Zndbstar} to get
\begin{equation}
\label{eq:lowboundZ}
\bbZ_n^{*,\gd,\gb(\gd)} \geq  e^{(\gep_{\gd} +k_1 \gd^4) n}\, 
E\bigg[ \exp\bigg\{-\tfrac{1}{4}(1+\eta)\gd^4 \sum_{x\in\bbZd} \ell_n(x)^2 \bigg\}
\ind_{\big\{\max_{x \in \bbZd} \ell_n(x)\leq a\, \gd^{-2}\big\}}\bigg].
\end{equation}
Fix $\eta \in (0,1)$ and pick $u= \tfrac{1}{4}(1+\eta)\gd^4$. Fix $\gep>0$ small, 
choose $\gep_\gd$ in \eqref{gbgddef} such that
\begin{equation}
\label{epschoice1}
\gep_{\gd}+k_1 \gd^4 = (1+\gep) f^\wsaw(u),
\end{equation}
and use \eqref{eq:lowboundZ} to estimate (recall \eqref{Qndef})
\begin{equation}
\label{lbestest}
\bbZ_n^{*,\gd,\gb(\gd)} \geq  e^{(1+\gep)f^\wsaw(u) n} E \left[ e^{-u Q_n} \ind_{\cE_n(u)} \right]
\end{equation}
with
\begin{equation}
\label{eventdef}
\cE_n(u) = \left\{\max_{x\in\bbZd} \ell_n(x) \leq c /\sqrt{u}\right\}
\end{equation}
and $c=a\tfrac12\sqrt{1-\eta}$. Below we prove the following lemma. 

\begin{lemma}
\label{lem:lbcritcurve}
For every $c>0$, $\gep>0$ and $0<u \leq u_0=u_0(c,\gep)$,
\begin{equation}
\label{claimsmallu}
\liminf_{n\to\infty} \frac1n \log E \left[ e^{-u Q_n} \ind_{\cE_n(u)} \right] 
\geq -(1+\tfrac12\gep) f^\wsaw(u).
\end{equation}
\hfill \qed
\end{lemma}

\noindent
Lemma~\ref{lem:lbcritcurve} in combination with \eqref{lbestest} implies that, for $\gd$ 
small enough,
\begin{equation}
F^*(\gd,\gb(\gd)) = \limsup_{n\to\infty} \frac{1}{n} \log \bbZ_n^{*,\gd,\gb(\gd)}
\geq \tfrac12\gep f^\wsaw(u) > 0  
\end{equation}
and hence $\gb_c(\gd) > \gb(\gd)$. But, by \eqref{gbgddef} and Proposition~\ref{pr:wsaw},
\begin{equation}
\gb(\gd) = \tfrac12\gb^2-\tfrac13 m_3\gd^3 - \gep_\gd,
\qquad
\gep_\gd = -k_1\gd^4 + (1+\gep) \left\{\begin{array}{ll}
\lambda_2 u\log(1/u), &d=2,\\[0.2cm]
\lambda_d u, &d \geq 3.
\end{array}
\right.
\end{equation}
Inserting $u= \tfrac{1}{4}(1+\eta)\gd^4$ into the last formula, we find that
\begin{equation}
\gep_\gd = [1+o_\gd(1)]\,\gd^4 \left\{\begin{array}{ll}
\tfrac14(1+\gep)(1+\eta)\lambda_2 \log(1/\gd), &d=2,\\[0.2cm]
\tfrac14(1+\gep)(1+\eta)\lambda_d-k_1, &d \geq 3. 
\end{array}
\right.
\end{equation}
Let $\eta,\gep \downarrow 0$ and recall \eqref{eq:k1def} to get the lower bound in 
\eqref{betacexp}. In the remainder of this section we prove Lemma~\ref{lem:lbcritcurve}. 

\begin{proof}
Without $\ind_{\cE_n(u)}$, the $\liminf$ is a $\lim$ and equals $-f^\wsaw(u)$. We must 
therefore show that the indicator does not change the free energy significantly. 

\medskip\noindent
$\bullet$ $d \geq 3$. The proof comes in 4 Steps. 

\medskip\noindent
{\bf 1.}
Recall \eqref{bridgedef}. We use the same idea as in the proof of Proposition~\ref{pr:largedevsuml2} 
(recall \eqref{eq:blockbridge}--\eqref{liminfPQn}), to write
\begin{equation}
E\left[ e^{-u Q_n} \ind_{\cE_n(u)}\right] \geq E \left[e^{-u Q_m} \ind_{\{\cE_m(u), 
S\in\cB_m\}} \right]^{n/m}, \qquad m \in N, \,n \in m\N.
\end{equation}
Choose
\begin{equation}
\label{eq:choicemu}
m = m(u) = \left\lceil \frac{\log^2(1/u)}{u} \right\rceil,
\end{equation}
so that $u \sim \frac{\log^2 m}{m}$ as $u \downarrow 0$, and $\cE_m(u) \supset \cE_m' 
= \{\sup_{x\in\Z^d} \ell_m(x) \leq  \frac{\sqrt{m}}{\log m}\}$ for $u$ small enough. We therefore 
get
\begin{equation}
E\left[e^{-u Q_n} \ind_{\cE_n(u)}\right] \geq E \left[e^{-u Q_m} \ind_{\{\cE_m', \cB_m\}}\right]^{n/m} 
=  \Big(P(\cE_m',\cB_m)\,E[e^{-u Q_m} \mid \cE_m', \cB_m]\Big)^{n/m}.
\end{equation}
Combining this inequality with Jensen's inequality, we obtain
\begin{equation}
\label{eq:lowerbound_uQn_Enu}
\liminf_{n\to\infty}\frac{1}{n} \log E\left[e^{-u Q_n} \ind_{\cE_n(u)}\right]
\geq \frac{1}{m} \log P(\cE_m',\cB_m) - \frac{u}{m} E[Q_m \mid \cE_m', \cB_m]. 
\end{equation}

\medskip\noindent
{\bf 2.}
Let us assume for the moment that
\begin{equation}
\label{eq:PEmBm}
\lim_{m\to\infty} P(\cE_m' \mid \cB_m) = 1
\end{equation}
and
\begin{equation}
\label{eq:EQmBm}
E[Q_m \mid \cB_m] \leq \gl_d m\,[1+o(1)], \qquad m\to\infty.
\end{equation}
Combining \eqref{eq:asymp.bridge} and \eqref{eq:lowerbound_uQn_Enu}--\eqref{eq:EQmBm}, 
we get
\begin{equation}
\liminf_{n\to\infty} \frac{1}{n} \log  \left[ e^{-u Q_n} \ind_{\cE_n(u)} \right] 
\geq - C\,\frac{\log m}{m} - [1+o(1)]\,\lambda_d u.
\end{equation}
From \eqref{eq:choicemu}, we have $\frac{\log m}{m} \sim \frac{u}{\log(1/u)}=o(u)$, 
$u \downarrow 0$. Therefore
\begin{equation}
\liminf_{n\to\infty} \frac{1}{n} \log E\left[ e^{-u Q_n} \ind_{\cE_n(u)}\right] 
\geq - [1+o(1)]\,\lambda_d u.
\end{equation}
Since $f^\wsaw(u) \sim \lambda_d u$, $u \downarrow 0$, by Proposition~\ref{pr:wsaw}, 
the claim in \eqref{claimsmallu} follows. 

\medskip\noindent
{\bf 3.} The claim in \eqref{eq:PEmBm} 
holds because
\begin{equation}
\begin{aligned}
P(\cE_m'^c  \mid \cB_m ) &\leq \frac{P\big(\exists\, x \in \Z^d\colon\, \ell_m(x) 
\geq \frac{\sqrt{m}}{\log m}\big)} {P(\cB_m)}\\
&\leq Cm^2 P\Big(\ell_\infty(0) \geq \frac{\sqrt{m}}{\log m} \Big)
\leq Cm^2 \exp\left(-C\,\frac{\sqrt{m}}{\log m}\right),
\end{aligned}
\end{equation}
where $\ell_\infty(0) = \lim_{m\to\infty} \ell_m(0)$, in the second inequality we use 
\eqref{eq:asymp.bridge} plus the fact that the range of simple random walk a time 
$m$ is at most $m$, and in the third inequality we use that simple random walk is 
transient. 

\medskip\noindent
{\bf 4.}
The claim in \eqref{eq:EQmBm} is proven in Appendix \ref{app:bridge2}.

\medskip\noindent
$\bullet$ $d=2$. Let $t_u=(1+\tfrac14\gep) \gl_2 \log(1/u)$, and estimate
\begin{equation}
E \left[ e^{-u Q_n} \ind_{\cE_n(u)} \right] \geq e^{-ut_u n} 
P\left(Q_n \leq t_u n,\,\cE_n(u)\right).
\end{equation}
As shown in Remark~\ref{rem:adaptdownard}, for $u$ small enough (i.e., for $t_u$ large enough)
\begin{equation}
\label{eq:dowarddevConstrained}
\lim_{n\to\infty} \frac1n \log P\big(Q_n \leq t_u n,\,\cE_n(u)\big) 
\geq - \exp\big( - (1+\tfrac14\gep)^{-1} t_u /\gl_2 \big) =-u.
\end{equation}
Hence
\begin{equation}
\liminf_{n\to\infty} \frac1n \log E \left[ e^{-u Q_n} \ind_{\cE_n(u)} \right] 
\geq -(1+\tfrac14\gep) \gl_2 u \log(1/u) -u \geq -(1+\tfrac12 \gep) f^\wsaw(u),
\end{equation}
where the last inequality is valid for $u$ small enough by Lemma \ref{pr:wsaw}. So, again,
the claim in \eqref{claimsmallu} holds.
\end{proof}

\subsection{Upper bound on the critical curve for small charge bias}
\label{ss:proofscalcritsmall-ub}

In this section we prove the upper bound in Theorem~\ref{thm:criticalpoint}(ii). Substitute 
\eqref{eq:upboundG1} into \eqref{Zndbstar} to get
\begin{equation}
\label{eq:upperboundZ}
\bbZ_n^{*,\gd,\gb(\gd)} \leq  
E\bigg[ \exp\bigg(\sum_{x\in\bbZd} \bigg\{ -\tfrac{1}{4}(1-\eta)\gd^4 
\ell_n(x)^2 + (\gep_{\gd} +k_1 \gd^4)  \ell_n(x)\bigg\}
\ind_{\{\ell_n(x)\leq a\, \gd^{-2} \}}  \bigg) \bigg].
\end{equation}
Fix $\eta \in (0,1)$ and choose $\gep_\gd$ in \eqref{gbgddef} such that
\begin{equation}
\label{epschoice2}
\gep_\gd + k_1 \gd^4 = \tfrac{1}{4}(1-\eta)\gd^4.
\end{equation} 
Using that $\ell(1-\ell) \le 0$ for all $\ell\in\N_0$, we readily get that $\bbZ_n^{*,\gd,\gb(\gd)} 
\le 0$. The upper bound for \eqref{betacexp} follows by noting that $\eta$ may be chosen 
arbitrarily small.

\subsection{Towards the conjectured scaling of the critical curve for small charge bias}
\label{s:toconj_sharpbetac}

In this section we state a technical property (Conjecture~\ref{conj:ubcritcurve} below)
that would imply the upper bound in Theorem~\ref{thm:criticalpoint}(ii) stated in 
Conjecture~\ref{conj:sharp_beta_c}. This property, in turn, would follow from a large 
deviation property of the trimmed range of simple random walk that we discuss in 
Appendix~\ref{app:range}.

Let us start from~\eqref{eq:upperboundZ}. Fix $\eta \in (0,1)$  and pick $u=\frac14 
(1-\eta)\gd^4$. Fix $\gep> 0$ small, choose $\gep_\gd$ in \eqref{gbgddef} such that
\begin{equation}
\label{epschoice2alt}
\gep_\gd + k_1 \gd^4 = (1-\gep) f^\wsaw(u),
\end{equation} 
and use \eqref{eq:upperboundZ} to estimate (recall \eqref{Qndef})
\begin{equation}
\label{defZbarest}
\bbZ_n^{*,\gd,\gb} \leq \bar Z_{n,u}^{\gep}
\end{equation}
with
\begin{equation}
\label{defZbar}
\bar Z_{n,u}^{\gep} = E\bigg[ \exp\bigg(\sum_{x\in\bbZd} 
\bigg\{ -u \ell_n(x)^2 + (1-\gep)f^\wsaw(u)  \ell_n(x)\bigg\}
\ind_{\{\ell_n(x)\leq 1/\sqrt{u} \}}  \bigg) \bigg].
\end{equation}
The following conjecture yields the sharp version of the upper bound missing in 
Theorem~\ref{thm:criticalpoint}(ii) via an argument similar to the one given below 
Lemma~\ref{lem:lbcritcurve}. 

\begin{conjecture}
\label{conj:ubcritcurve}
For every $\gep>0$ and $0<u\leq u_0(\gep)$,
\begin{equation}
\label{targetpropZ} 
\limsup_{n\to\infty} \frac1n \log \bar Z_{n,u}^{\gep} = 0.
\end{equation}
\hfill \qed
\end{conjecture}

\subsection{Scaling of the critical curve for large charge bias}
\label{ss:proofscalcritlarge}

Theorem~\ref{thm:criticalpoint}(iii) is the same as for $d=1$ in \cite{CdHPP16}, and 
the proof carries over verbatim.

\section{Scaling of the annealed free energy}
\label{s:scalingfreeenergy}


\subsection{Scaling bounds on the annealed free energy for small inverse temperature}

In this section we prove Theorem~\ref{thm:hightemp}.

\begin{proof}
The proof is based on Proposition~\ref{pr:wsaw} and proceeds via lower and upper bounds.
The upper bound uses a uniform upper bound for $g_{\gd,\gb}$ defined in \eqref{gdbdef} 
for small $\gb$ (Lemma~\ref{lem:gdb} below).

\bigskip\noindent
{\bf Lower bound:}
Jensen's inequality applied to \eqref{eq:def.annpartfunc}--\eqref{eq:def.hamiltonian} 
gives
\begin{equation}
\begin{aligned}
\Z_{n}^{\gd,\gb} 
&= \bbE^{\gd} \bigg[E\bigg[\exp\Big( -\gb \sum_{1\leq i,j \leq n} \go_i \go_j  
\ind_{\{ S_i=S_j\}}\Big) \bigg] \bigg]\\ 
&\geq E \bigg[ \exp\Big( -\gb \sum_{1\leq i,j \leq n} \bbE^{\gd}[\go_i \go_j ] 
\ind_{\{ S_i=S_j\}}\Big)\bigg]\\
&= e^{-n\gb v(\gd)} E\left[ \exp\Big( -\gb m(\gd)^2 \sum_{1\leq i,j \leq n} 
\ind_{\{ S_i=S_j\}}\Big)\right]= e^{-n\gb v(\gd)} E\left[e^{-\gb m(\gd)^2 Q_n}\right],
\end{aligned}
\end{equation}
where we recall that $m(\gd) = \bbE^\gd[\go_1]$ and $v(\gd) = \bbVar^\gd[\go_1]$. 
Hence
\begin{equation}
F(\gb,\gd) \geq - f^{\wsaw} \big(\gb m(\gd)^2 \big) -\gb v(\gd).
\end{equation}
Use Proposition~\ref{pr:wsaw} to get the lower bound in \eqref{eq:Fscal}.

\bigskip\noindent
{\bf Upper bound:}
Recall \eqref{Zndbdef}--\eqref{gdbdef}. We need the following lemma.
 
\begin{lemma}
\label{lem:gdb}
For every $\eta>0$ there exist $a=a(\eta)>0$ and $\gb_0=\gb_0(\eta)>0$ such that 
the following hold for all $\gb\leq \gb_0$.\\
{\rm (1)} If $\gb \ell^2\leq a $, then
\begin{equation}
\label{gendest1}
g_{\gd,\gb}(\ell) \leq \exp\Big(-\big[\gb v(\gd)\ell+(1-\eta)\gb m(\gd)^2\ell^2\big]\Big)
\qquad \forall\,\gd>0.
\end{equation}
{\rm (2)} There exists a constant $c_{\gd}>0$ (depending only on $\gd$) such that 
if $\gb\ell^2>a$, then
\begin{equation}
\label{gendest2}
g_{\gd,\gb}(\ell) \leq \exp \Big(-c_\gd \min \{\gb\ell^2,\ell\}\Big) \qquad \forall\,\gd>0. 
\end{equation}
\hfill \qed
\end{lemma}

\begin{proof}
For the case $ \gb^2 \ell\leq a$, we use that $e^{-t} \leq 1-t +t^2$, $t\geq 0$, 
to estimate
\begin{equation}
\begin{aligned}
g_{\gd,\gb}(\ell) &\leq 1- \gb\bbE^\gd[\gO_\ell^2] + \gb^2 \bbE^\gd[\gO_\ell^4]  
\leq 1-\gb \big( m(\gd)^2 \ell^2 + v(\gd) \ell\big) + c \gb^2 \ell^4\\
&\leq 1-\gb \big( m(\gd)^2 \ell^2 + v(\gd) \ell\big) + \eta^2 \gb \ell^2  
\leq \exp\Big(- \big[\gb v(\gd) \ell + (1-\eta)\gb m(\gd)^2  \ell^2\big] \Big),
\end{aligned}
\end{equation}
where we use that $\gb \ell^2 \leq a$, with $a$ chosen small enough so that $ca \leq \eta^2$.

For the case $\gb\ell^2 > a$, we estimate
\begin{equation}
\label{gestell}
g_{\gd,\gb}(\ell) \leq e^{-\gb \tfrac12 m(\gd)^2 \ell^2 } 
+ \bbP^\gd\big( \gO_\ell^2 \leq \tfrac12 m(\gd)^2 \ell^2\big) .
\end{equation}
For the last term we can use the large deviation principle for $\gO_\ell$: since $\ell 
> \sqrt{a/\gb} \gg 1$, there exists a rate function $J$, with $J(t)>0$ for $0<t<m(\gd)$, 
such that $\bbP^{\gd}(\gO_\ell \leq t \ell) \leq e^{-J(t)\ell}$. Hence \eqref{gestell} gives 
\begin{equation}
g_{\gd,\gb}(\ell) \leq e^{-\gb \tfrac14 m(\gd)^2\ell^2} + e^{-J\big(\tfrac12 m(\gd)\big)\ell}.
\end{equation}
We next use that either $\tfrac14 m(\gd)^2\gb \ell^2 \leq 1 \ll J\big(\tfrac12 m(\gd)\big)\ell$ 
or both $\tfrac14 m(\gd)^2\gb \ell^2$ and $J\big(\tfrac12 m(\gd)\big)\ell$ are $\geq 1$, to 
get that there is a constant $c>0$ such that
\begin{equation}
g_{\gd,\gb}(\ell) \leq  \max\Big\{e^{- c m(\gd)^2\gb\ell^2},
\,e^{- c J\big(\tfrac12 m(\gd)\big)\ell}\Big\},
\end{equation}
which proves the claim with $c_\gd=\max\{c m(\gd)^2 , cJ(\tfrac12 m(\gd))\}$.
\end{proof}

With the help of Lemma~\ref{lem:gdb} we can now prove the upper bound. Inserting 
\eqref{gendest1}--\eqref{gendest2} into \eqref{Zndbdef}, we get the upper 
bound
\begin{equation}
\label{ubhelp}
\begin{aligned}
\bbZ_{n}^{\gd,\gb} &\leq E\bigg[\exp\bigg(-\sum_{x\in\bbZ^d} 
\Big\{\Big[ \gb v(\gd) \ell_n(x) + (1-\eta)\gb m(\gd)^2\ell_n(x)^2\Big] 
\ind_{\{\ell_n(x)\leq a\gb^{-1/2}\}}\\ 
&\qquad\qquad\qquad\qquad\qquad 
+ \left[c_\gd \min \big\{\gb \ell_n(x)^2,\ell_n(x)\big\}\right] 
\ind_{\{\ell_n(x)>a \gb^{-1/2}\}}\Big\}\bigg) \bigg].
\end{aligned}
\end{equation}
Let $u = (1-\eta)\gb m(\gd)^2$. Then the condition $\ell_n(x)\leq a\gb^{-1/2}$ translates into 
$\ell_n(x)\leq c_{\gd,\eta} /\sqrt{u}$, and for any $\gep>0$ the upper bound in \eqref{ubhelp} 
gives
\begin{equation}
\label{finest1}
\begin{aligned}
\bbZ_{n}^{\gd,\gb} &\leq e^{-\gb v(\gd) n - un}\\ 
&\qquad\times
E \bigg[\exp\bigg(\sum_{x\in\bbZ^d} \Big\{-u \ell_n(x)^2 + u\ell_n(x)\Big\} 
\ind_{\{\ell_n(x)\leq c/\sqrt{u}\}} \bigg)\\
&\qquad\qquad\qquad\times \exp\bigg(\sum_{x\in\bbZ^d} h_{\gd,\gb}(\ell_n(x))
\ind_{\{\ell_n(x)> c/\sqrt{u}\}}\bigg) \bigg]
\end{aligned}
\end{equation}
with
\begin{equation}
h_{\gd,\gb}(\ell) = -c_{\gd}\min \big\{\gb \ell^2,\ell\big\} + \gb v(\gd) \ell + (1-\gep)f^{\wsaw}(u) \ell.
\end{equation}
Since $\ell(1-\ell) \le 0$ for all $\ell\in\N_0$, we get
\begin{equation}
\bbZ_{n}^{\gd,\gb} \leq e^{-\gb v(\gd) n - un}
E \bigg[\exp\bigg(\sum_{x\in\bbZ^d} h_{\gd,\gb}(\ell_n(x))
\ind_{\{\ell_n(x)> c/\sqrt{u}\}}\bigg) \bigg]
\end{equation}
However, $h_{\gd,\gb}(\ell) \leq 0$ when $\gb$ is small enough and $\ell > a \gb^{-1/2}$ (or 
$\ell > c/\sqrt{u}$). Indeed, using that $f^{\wsaw}(u) = o (\gb^{1/2})$ as $\gb\downarrow 0$ 
by Proposition~\ref{pr:wsaw}, we get, as $\beta \downarrow 0$,
\begin{equation}
\label{finest2}
h_{\gd,\gb}(\ell) \leq 
\begin{cases}
[-c_{\gd} + \gb v(\gd)+f^{\wsaw}(u)]\ell = -[1+o(1)] c_{\gd} \ell, 
&\ell \geq 1/\gb,\\
[-c_{\gd} a \gb^{1/2} + \gb v(\gd)+f^{\wsaw}(u)]\ell = -[1+o(1)] c_{\gd} a^2, 
&a\gb^{-1/2} \leq \ell < 1/\gb.
\end{cases}
\end{equation}
Finally, we get $\bbZ_{n}^{\gd,\gb} \leq e^{-\gb v(\gd) n - un}$, which gives the upper bound.
\end{proof}

\subsection{Towards the conjectured scaling of the free energy for small inverse temperature}

In this section we explain how to settle Conjecture~\ref{conj:asym.smallbeta} with the help of 
Conjecture~\ref{conj:ubcritcurve}. Instead of \eqref{finest1}, we write 
\begin{equation}
\label{finest1bis}
\begin{aligned}
\bbZ_{n}^{\gd,\gb} &\leq e^{-\gb v(\gd) n - (1-\gep) f^{\wsaw}(u) n}\\ 
&\qquad\times
E \bigg[\exp\bigg(\sum_{x\in\bbZ^d} \Big\{-u \ell_n(x)^2 +(1-\gep) f^{\wsaw}(u)\ell_n(x)\Big\} 
\ind_{\{\ell_n(x)\leq c/\sqrt{u}\}} \bigg)\\
&\qquad\qquad\qquad\times \exp\bigg(\sum_{x\in\bbZ^d} h_{\gd,\gb}(\ell_n(x))
\ind_{\{\ell_n(x)> c/\sqrt{u}\}}\bigg) \bigg]
\end{aligned}
\end{equation}
Combining \eqref{finest1bis} and \eqref{finest2}, and recalling \eqref{defZbarest}--\eqref{defZbar}, 
we get
\begin{equation}
\bbZ_{n}^{\gd,\gb} \leq e^{-\gb v(\gd)n - (1-\gep) f^{\wsaw}(u) n} \bar Z_{n,u}^{\gep}.
\end{equation}
Because of \eqref{targetpropZ}, we find that $\limsup_{n\to\infty} \frac{1}{n} \log \bar Z_{n,u}^{\gep} 
=0$ for any $\gep>0$, provided $u$ is small enough (i.e., provided $\gb$ is small enough).
Since $u=(1-\eta) \gb m(\gd)^2$, we conclude that, for any fixed $\eta,\gep>0$,
\begin{equation}
F(\gd,\gb) = \limsup_{n\to\infty} \frac{1}{n} 
\log \Z_{n}^{\gd,\gb} \leq -\gb v(\gd) - (1-\gep) f^{\wsaw}((1-\eta) \gb m(\gd)^2 ).
\end{equation}
Let $\gep,\eta \downarrow 0$ to get the upper bound in \eqref{eq:Fscal}.

\section{Super-additivity for large inverse temperature}
\label{s:conv_free_energy}

In this section we prove Theorem~\ref{thm:conv_free_energy}. Looking back at \eqref{Zndbstar}, 
we first note that item (1) combined with
\be
\ell_{n+m}(x) = \ell_n(x) + \ell_{(n,n+m]}(x),\qquad \ell_{(n,n+m]}(x) = \sum_{n<k\le n+m} \ind_{\{S_k = x\}},
\ee
and
\be
\bE\Big( \prod_x g^*_{\gd,\gb}(\ell_{(n,n+m]}(x))\ \Big| S_0, \ldots, S_n\ \Big) = \bbZ_m^{*,\gd,\gb}
\ee
implies that the annealed partition function is super-multiplicative, which yields items (2) and (3). 

We next prove item (1). The proof consists of a refinement of the proof of 
Theorem~\ref{thm:criticalpoint}(iii). Recall that 
\begin{equation}
g^*_{\gd,\gb}(\ell) = \bbE (e^{-\gb \gO_\ell^2 + \gd\gO_\ell}). 
\end{equation}
In the following we will denote by $f_\ell$ the density of $\gO_\ell$, and use that
\begin{lemma}
\label{lem:estimates_fl}
There exist $\gep_0 >0$ and two positive constants $c_0$ and $c_1$ such that for $\ell \ge 1$,
\begin{equation}
c_0\ \ell^{-1/2} \le \inf_{0\le x\le \gep_0} f_\ell(x) \le \|f_\ell\|_\infty \le c_1\ \ell^{-1/2}.
\end{equation}
\end{lemma}
We will also use the following estimates on the function $g^*_{\gd,\gb}$:
\begin{lemma}
\label{lem:estimates_gstar}
Suppose that $\gb(\gd)$ is such that $\gd  \ll \gb(\gd)) \ll \gd^2$ as $\gd\to\infty$. Then, there 
exists a constant $c>1$ such that for $\gd$ large enough, $\ell\in\N$, $\eta\in (0,1)$
\begin{equation}
(1/c) \eta \frac{\gd}{\gb(\gd)}\ e^{(1-\eta)\gd^2/4\gb(\gd)}\ \ell^{-1/2} 
\le g^*_{\gd,\gb(\gd)}(\ell) \le c\ e^{\gd^2/4\gb(\gd)}\frac{\gd}{\gb(\gd)}\ \ell^{-1/2}.
\end{equation}
\end{lemma}
Using the previous lemma we get, for some constant $c>0$, $\eta\in(0,1)$ and all $m,n\in\N$,
\begin{equation}
\label{eq:superadd}
\begin{aligned}
&\log g^*_{\gd,\gb}(m+n) - \log g^*_{\gd,\gb}(m) - \log g^*_{\gd,\gb}(n)\\
&\ge \tfrac12 \inf_{u,v\ge 1} \{\log u + \log v - \log(u+v)\} - c + \log \eta 
+ \big[ \log(\gb/\gd) - (1+\eta)\tfrac{\gd^2}{4\gb}  \big].
\end{aligned}
\end{equation}
Picking for $\gb$ the value $\gb(\gd) = (1+\sqrt{\eta})\frac{\gd^2}{4\log \gd}$ with $\eta\in (0,1)$, 
the right-hand side of \eqref{eq:superadd} becomes positive for $\gd$ large enough, which 
proves item (1). Note that this value of $\gb(\gd)$ satisfies the assumption of
 Lemma~\ref{lem:estimates_gstar} and is equivalent to $(1+\sqrt\eta)\gb_c(\gd)$, in view 
 of Theorem~\ref{thm:criticalpoint}(iii). Since $\eta$ can be made arbitrarily small, this 
 completes the proof of the theorem.

\begin{proof}[Proof of Lemma~\ref{lem:estimates_fl}]
This follows from the local limit theorem for densities (see Petrov~\cite[Theorem 7, Chapter 
VII]{P75}), where we need that the density of $\omega_1$ is bounded.
\end{proof}

\begin{proof}[Proof of Lemma~\ref{lem:estimates_gstar}] 
In the following we pick $\gb(\gd)$ as in the statement of the lemma, but we write $\gb$ 
for simplicity. We start with the decomposition
\be
g^*_{\gd,\gb}(\ell) = \int_\R e^{\gd s (1 - \gb s/\gd)} f_\ell(s),\dd s = I_1 + I_2 + I_3,
\ee
where
\be
I_1 = \int_{\{0<s< \gd/\gb\}}, \quad
I_2 = \int_{\{-\gep<s< 0\} \cup \{\gd/\gb < s < \gd/\gb + \gep\}}, \quad
I_3 = \int_{\{s < -\gep \} \cup \{s > \gd/\gb + \gep\}},
\ee
and $\gep>0$ will be determined later. For the lower bound, we may write
\be
I_1 \ge \eta (\gd/\gb) e^{ \frac{\gd^2}{4\gb} (1-\eta)  } \inf_{\tfrac{\gd}{2\gb}< s < (1+\eta) \tfrac{\gd}{2\gb}} f_\ell(s)
\ee
and use Lemma~\ref{lem:estimates_fl}, since $\gd/\gb < \gep_0 /2$ for $\gd$ large enough.
For the upper bound, we easily get
\be
I_1 \le e^{\gd^2/4\gb} \frac{\gd}{\gb} \|f_\ell\|_\infty,\qquad I_2 \le 2\gep \|f_\ell\|_\infty.
\ee
As to the third term, we have
\be
I_3 \le \int_{s<-\gep} e^{\gd s} f_\ell(s)\,\dd s + \int_{s > \gd/\gb + \gep} 
e^{-\gb\gep s} f_\ell(s)\,\dd s \le \Big(\frac{1}{\gd} + \frac{1}{\gb\gep} \Big) \|f_\ell\|_\infty.
\ee
By picking $\gep = \gd/\gb$, we obtain
\be
g^*_{\gd,\gb}(\ell) \le e^{\gd^2/4\gb} \frac{\gd}{\gb}\, \|f_\ell\|_\infty\, (3 + 2\gb/\gd^2).
\ee
We can now complete the proof with the help of Lemma~\ref{lem:estimates_fl}, since the last expression 
in parenthesis is less than $4$ for $\gd$ large enough.
\end{proof}


\begin{appendix}

\section{Bridge estimates}
\label{app:bridge}

In this appendix we collect the estimates about simple random walk conditioned to 
be a bridge that were claimed in \eqref{eq:asymp.bridge}, \eqref{liminfPQnalt} and 
\eqref{eq:EQmBm}.

\subsection{Bridge probability}
\label{app:bridge0}

First we prove \eqref{eq:asymp.bridge}. Note that it suffices to give the proof for $d=1$. 
Indeed, by a standard large deviation estimate, the number of steps taken by the random
walk in direction $1$ after it has taken $n$ steps in total equals $\tfrac1d n[1+o(1)]$, with 
an exponentially small probability of deviation. Hence, if the claim is true for $d=1$, then 
it is also true for $d \geq 2$ with $C$ replaced by $dC$.

To prove the claim for $d=1$ we write
\begin{equation}
\label{br1}
\begin{aligned}
P(\cB_{2n}) &= \sum_{x=1}^\infty P(\cB_{2n}, S_{2n} = x)\\
&= \sum_{x=2}^\infty \sum_{y=1}^x 
P\left(S_n = y,\,\max_{0 < k < n} S_k < x, \min_{0<k<n} S_k >0\right)\\
&\qquad\qquad\qquad 
\times P\left(S_n = x-y,\,\max_{0 < k < n} S_k < x, \min_{0<k<n} S_k > 0\right),
\end{aligned}
\end{equation}  
where the product after the second equality arises after we use the Markov property at time 
$n$ and reverse time in the second half of the random walk. Let $(\alpha_n)_{n\in\N}$ and 
$(\beta_n)_{n\in\N}$ be sequences in $(0,\infty)$ that tend to $\alpha$ and $\beta$, respectively, 
with $0 \leq \beta \leq \alpha$. Then it follows from Caravenna and Chaumont~\cite[Theorem 2.4]{CC13} 
that
\begin{equation}
\label{br4}
\lim_{n\to\infty} P\left(\max_{0<k<n} S_k < \alpha_n\sqrt{n} ~\Big|~ S_n = \beta_n\sqrt{n}, 
\min_{0<k<n} S_k > 0\right) = \psi(\alpha,\beta)
\end{equation}
with
\begin{equation}
\label{br5}
\psi(\alpha,\beta) = P^*\left( \max_{0 \leq t \leq 1} X_t^\beta \leq \alpha\right).
\end{equation} 
Here, $(X_t^\beta)_{0 \leq t \leq 1}$ is the Brownian bridge between $0$ and $\beta$ conditioned 
to stay positive, and $P^*$ denotes its law. Moreover, by the ballot theorem (Feller~\cite{F71}), we have
\begin{equation}
\label{br6}
P\left(S_n = \beta_n\sqrt{n}, \min_{0<k<n} S_k > 0\right)
= \frac{\beta_n\sqrt{n}}{n}\,P\big(S_n = \beta_n\sqrt{n}\,\big),
\end{equation}
so that
\begin{equation}
\label{br7}
\lim_{n\to\infty} n\,P\left(S_n = \beta_n\sqrt{n}, \min_{0<k<n} S_k > 0\right) 
= \beta n(\beta)
\end{equation}
with $n(z) = \frac{1}{\sqrt{2\pi}} \exp[-\tfrac12z^2]$, $z\in\R$, the standard normal density.
Rewriting \eqref{br1} as
\begin{equation}
\label{br8}
\begin{aligned}
n\,P(\cB_{2n}) = \sum_{x=2}^\infty \sum_{y=1}^{x-1}
&\quad \frac{1}{\sqrt{n}}\,\,n\,P\left(S_n = y, \min_{0<k<n} S_k > 0\right)\\
&\qquad\qquad\qquad \times
P\left(\max_{0<k<n} S_k < x ~\Big|~ S_n = y,\min_{0<k<n} S_k > 0\right)\\
&\times \frac{1}{\sqrt{n}}\,\,n\,P\left(S_n = x-y, \min_{0<k<n} S_k > 0\right)\\
&\qquad\qquad\qquad \times
P\left(\max_{0<k<n} S_k < x ~\Big|~ S_n = x-y,\min_{0<k<n} S_k > 0\right),
\end{aligned}
\end{equation}
changing variables $x=\alpha_n\sqrt{n}$ and $y = \beta_n\sqrt{n}$, and taking the limit $n\to\infty$, 
we get with the help of \eqref{br4}, \eqref{br6} and \eqref{br7} that
\begin{equation}
\lim_{n\to\infty} n\,P(\cB_{2n}) = C'
\end{equation}
with
\begin{equation}
C' = \int_0^\infty d\alpha \int_0^\alpha d\beta\,\,
\big[\beta n(\beta)\,\psi(\alpha,\beta)\big]\,
\big[(\alpha-\beta) n(\alpha-\beta)\,\psi(\alpha,\alpha-\beta)\big].
\end{equation}
The limit and the integral can be interchanged with the help of dominated convergence (drop the
two conditional probabilities in \eqref{br8} and write the resulting bound as the square of $\sqrt{n} 
\,P(\min_{0<k<n} S_k > 0)$, which tends to $1/\sqrt{2\pi}$ as $n\to\infty$). The same argument 
works for $P(\cB_{2n+1})$ after cutting at time $n$, which leads to two random walks of length 
$n$ and $n+1$, but yields the same asymptotics. 

Thus, we have proved \eqref{eq:asymp.bridge} for arbitrary $d \geq 1$ with $C = 2dC'$. It is possible 
to derive a closed form expression for $\psi(\alpha,\beta)$ because $(X_t^\beta)_{0 \leq t \leq 1}$ 
is a $\beta$-dependent Doob-transform of Brownian motion. However, the value of $C'$ is of no concern 
to us. Note that  
\begin{equation}
0 < C' < \int_0^\infty d\alpha \int_0^\alpha d\beta\,\,
\big[\beta n(\beta)\big]\,\big[(\alpha-\beta) n(\alpha-\beta)\big]
= \left(\int_0^\infty d\gamma\,\gamma n(\gamma) \right)^2 = \frac{1}{2\pi}.
\end{equation}   
 
\subsection{Self-intersection local time for bridges in dimension two}
\label{app:bridge1}

We next prove \eqref{liminfPQnalt}. The idea is that the main contribution comes from the 
restriction $S_{[0,m]}\in\cB_m$. Fix $\gep>0$ small, let $t_m= \gep^2 m$, and consider the 
three time intervals $I_1=(1,t_m]$, $I_2=(t_m,m-t_m]$, $I_3=(m-t_m,m]$. Define $Q^{k,l} 
= \sum_{i\in I_k , j\in I_l} \ind_{\{S_i=S_j\}}$, $k,l\in\{1,2,3\}$ (so that $Q_m=\sum_{k,l
\in\{1,2,3\}} Q^{k,l}$), and define the events
\begin{equation}
\begin{aligned}
\cD_{k,l} &= \{ Q^{k,l} \leq \tfrac{\gep}{100} m \log m \}, \quad (k,l)\neq (2,2),\\
\cD_{2,2} &= \{Q^{2,2} \leq (1+\gep/2) \gl_2 m\log m\}.
\end{aligned}
\end{equation}
Then, provided $\gep$ is small enough, we have
\begin{equation}
\begin{aligned}
&P\big(Q_m  \leq (1+\gep)\gl_2 m \log m,\, \cB_m\big)\\
&\geq P\big( \cD_{k,l} \  k,l\in\{1,2,3\} ,\, \cB_m \big)
\geq  P(\cB_m)\,\bigg[1 - \sum_{k,l \in\{1,2,3\}} P(\cD_{k,l} ^c \mid \cB_m)\bigg],
\label{eq:Qin3}
\end{aligned}
\end{equation}
where we use the union bound, and the notation $\cB_m$ is short for $S_{[0,m]}\in\cB_m$.
We claim that, for $m$ large enough,
\begin{equation}
\label{smallprobab}
P(\cD_{k,l} ^c \mid \cB_m) \leq 100 \gep, \qquad k,l \in\{1,2,3\},
\end{equation}
which in turns proves \eqref{liminfPQnalt} because $\gep$ is arbitrary.

The proof of \eqref{smallprobab} goes as follows. First consider $(k,l)\neq (2,2)$. The 
Markov inequality gives
\begin{equation}
P(\cD_{k,l} ^c \mid \cB_m) \leq \frac{100}{\gep m\log m} E[Q^{k,l} \mid \cB_m],
\end {equation}
and so we need to estimate the last term. By symmetry, we may deal with the case 
$k=1$ only. Write 
\begin{equation}
E[Q^{1,l} \mid \cB_m] \leq t_m + 2 \sum_{i\in I_1} \sum_{j=i+1}^{m} P( S_i=S_j \mid \cB_m).
\end{equation}
Using the Markov property at times $i$ and $j$ and setting $r=j-i$, we get
\begin{equation}
\label{eq:conditionedbridge}
\begin{aligned}
&P(S_i=S_j, \cB_m)  = \sum_{x\in\bbZ^d} \sum_{y > x_1} P\left(S_i=S_j=x,\,
S_m^{(1)}=y,\, 0<S_k^{(1)} < y\,\, \forall\,0 < k < m\right)\\
&\leq \sum_{x\in\Z^d} \sum_{y > x_1} P\left(S_i=S_j=x,\,S_m^{(1)}=y,
0<S_k^{(1)} < y\,\, \forall\,0 < k < i\,\,\forall\, j < k < m\right)\\
&= \sum_{x\in\Z^d} \sum_{y > x_1} 
P\left(S_i=x,\,0<S_k^{(1)} < y\,\,\forall\,0 < k < i \right) P(S_{j-i} =0)\\
&\qquad \qquad \times 
P\left(S^{(1)}_{m-j}=y,\,0<S_k^{(1)} < y\,\, \forall\, 0 < k < m-j \mid S_0=x \right) \\
&= P(S_r =0) \sum_{x\in\Z^d} \sum_{y > x_1} P\left(S_i=x, \,S^{(1)}_{i+m-j}=y,\,
0<S_k^{(1)} < y\,\, \forall\, 0 < k < i+ m-j \right)\\
&= P(S_r =0)\, P(\cB_{m-r}).
\end{aligned}
\end{equation}
Hence, using the local limit theorem to get that there is a constant $c>0$ such that 
$P(S_r =0) \leq \frac{c}{r+1}$, and also \eqref{eq:asymp.bridge} to obtain the bound 
$P(\cB_{m-r}) \big/ P(\cB_m) \leq c \frac{m}{m-r}$, we get that
\begin{equation}
E[Q^{1,l} \mid \cB_m] \leq t_m + 2 c^2\sum_{i\in I_1} 
\sum_{r =1}^{m-i} \frac{1}{r+1} \frac{m}{m-r} \leq c' t_m \log m.
\end{equation}
Therefore, thanks to the definition of $t_m$, we get that
\begin{equation}
P(\cD_{k,l} ^c \mid \cB_m) \leq 100 \gep  \qquad \text{ for } (k,l)\neq (2,2).
\end{equation}

It remains to deal with the case $k=l=2$. We use \eqref{eq:asymp.bridge} to get that there 
is a constant $c>0$ such that
\begin{equation}
\begin{aligned}
P(\cD_{2,2} ^c \mid \cB_m) 
&\leq c \, m P\big( Q^{2,2}> (1+\gep/2) m\log m,\cB_m \big)\\
&\leq c\, m P\big( Q^{2,2}> (1+\gep/2) m\log m,\, S_i^{(1)}>0 \,\, 
\forall\, i\in I_1 , S_i^{(1)} < S_m^{(1)} \,\, \forall\, i \in I_3 \big)\\
&\leq  c\, m P \big( S_i^{(1)} > 0\,\, \forall\, 0 < i \leq \gep^2 m \big)^2 
P\big( Q_{(1-2\gep^2)m}> (1+\gep/2) m\log m \big)\\
&\leq  \frac{c'}{\gep^2} P\big( Q_{(1-2\gep^2)m}> (1+\gep/2) m\log m \big),
\end{aligned}
\end{equation}
where we use the independence of the three events in the second inequality, and the estimate 
$P(S_i^{(1)} > 0\,\,\forall\, 0< i \leq t\big) \leq c/\sqrt{t}$ in the third inequality. Finally, we simply use 
that $P(Q_{(1-2\gep^2)m}> (1+\gep/2) m\log m) \to 0$ as $m\to\infty$ (by a standard second 
moment estimate), so that \eqref{smallprobab} holds for large enough $m$.

\subsection{Self-intersection local time for bridges in dimensions three and higher}
\label{app:bridge2}

We finally prove \eqref{eq:EQmBm}. Recall from \eqref{gld} that $\gl_d = 2G_d-1 = 1 + 
2 \sum_{n\in\bbN} P(S_n=0)$.  We may write
\begin{equation}
\label{eq:Qmcond}
E[Q_m \mid \cB_m] = \sum_{1\leq i,j\leq m} P(S_j = S_i \mid \cB_m) 
\leq m + 2 \sum_{1\leq i < j \leq m} P(S_j - S_i = 0 \mid \cB_m)
\end{equation}
and use \eqref{eq:conditionedbridge}. By Remark~\ref{rem:bridge}, for every $\gep>0$ and 
$A<\infty$ there exists an $m_0=m_0(\gep,A)<\infty$ such that, for all $m \geq m_0$,
\begin{equation}
\label{eq:PcondSiSj}
P(S_j - S_i = 0 \mid \cB_m) \leq P(S_r=0)\,\frac{P(\cB_{m-r}) }{P(\cB_{m})}
\leq
\begin{cases}
(1+\gep) P(S_r = 0), &\mbox{ if }  1\leq r \leq A,\\
C\,P(S_r = 0), &\mbox{ if }  A < r \leq m/2,\\
C\,\frac{m^{1-d/2}}{1+ m-r}  &\mbox{ if }  m/2 < r \leq m,
\end{cases}
\end{equation}
where in the third line we use the standard local limit theorem to estimate $P(S_r =0) 
\leq Cm^{-d/2}$ for all $r \geq m/2$. Using \eqref{eq:PcondSiSj} we get, for any 
$1 \leq i \leq m$,
\begin{equation}
\label{Pijest}
\begin{aligned}
&\sum_{i<j \leq m} P(S_j-S_i=0 \mid \cB_m)\\
&\leq (1+\gep) \sum_{1 \leq r \leq A} P(S_r=0) 
+ C \sum_{A<r \leq m/2} P(S_r=0) 
+ Cm^{1-d/2} \sum_{m/2<r \leq m} \frac{1}{1+m-r} \\
&\leq (1+2\gep) \sum_{r \in \N} P(S_r=0) + Cm^{1-d/2} \log m \leq (1+3\gep) 
\sum_{r \in \N} P(S_r=0),
\end{aligned}
\end{equation}
where we use that $d\geq 3$, take $A$ large enough so that $C\sum_{r > A} P(S_r=0) 
\leq \gep \sum_{r \in\N} P(S_r=0)$, and take $m$ large enough. Substitute \eqref{Pijest}
into \eqref{eq:Qmcond} and sum over $1 \leq i \leq m$, to get 
\begin{equation}
E[Q_m \mid \cB_m] \leq (1+3\gep)\,m\, \Big(1+ 2\sum_{r \in \N} P(S_r=0)\Big) 
= (1+3\gep)\,\gl_d m,
\end{equation}
which concludes the proof.

\section{A conjecture for weakly self-avoiding walk}
\label{app:wsaw}

In this appendix we complement Proposition~\ref{pr:wsaw} by stating a conjecture
for the higher order terms in the asymptotic expansion of $f^{\wsaw}(u)$ for $d\ge 3$.

\begin{conjecture}
\label{conjfwsaw}
There are constants $a_d>0$ such that
\be
\lambda_d u  - f^{\wsaw}(u) \sim 
\begin{cases}
a_3 \, u^{3/2} ,& \quad d=3, \\
 a_4 \, u^2 \log(1/u), & \quad d=4, \\
a_d \, u^2, & \quad d\ge 5,
\end{cases}
\qquad \text{as } u\downarrow 0.
\ee
\end{conjecture}

\noindent
Via \eqref{eq:Varadhan} this translates into a related conjecture for the rate function 
$I$ in Proposition~\ref{pr:largedevsuml2}: we conjecture that there are constants $\tilde a_d >0$ such that
\begin{equation}
I(\lambda_d -s) \sim
\begin{cases}
\tilde a_3 \, s^{3}, & \quad d=3, \\
\tilde a_4 \, s^2/\log(1/s), & \quad d=4, \\
\tilde a_d \, s^2, & \quad d\ge 5,
\end{cases}
\qquad s\downarrow 0.
\end{equation}

\smallskip
Let us develop some heuristic arguments to support Conjecture~\ref{conjfwsaw}.
First of all, note that in dimension $d\ge 3$, there are constants $\tilde c_d$ such that
\begin{equation}
\label{EQn}
\lambda_d n - E[Q_n] \sim
\begin{cases}
\tilde c_{3} n^{1/2}, &d=3,\\
 \tilde c_4 \log n, &d=4,\\
\tilde c_d,  &d\ge 5,
\end{cases}
\qquad n\to\infty.
\end{equation}
Indeed, we may write
\be
Q_n = n+ 2\sum_{ i=1}^{ n-1} \sum_{j=i+1}^{n} \ind_{\{S_j=S_j\}} 
= n + 2 \sum_{ i=1}^{n-1} \Big(\sum_{k=1}^{\infty} \ind_{\{S_{i+k} =S_i\}} \Big)  
- 2\sum_{i=1}^{n-1} \sum_{j> n} \ind_{\{S_j = S_i\}},
\ee
so that, by taking the expectation, we get
\be
E[Q_n] = n+ 2 (n-1) G_d -  2 \sum_{i=1}^{n-1} \sum_{j> n} P(S_{j-i}=0).
\ee
The first term equals $\lambda_d n -2G_d$. The second term can be easily estimated: 
we have $P(S_{2k} =0) \sim (2/\pi)^{d/2} k^{-d/2}$ as $k \to\infty$, so that $\sum_{j> n} 
P(S_{j-i}=0) \sim  \frac{2^d}{\pi^{d/2}(d-2)} (n-i)^{1-d/2}$ as $n - i \to \infty$. Hence
\be
\sum_{i=1}^{n-1} \sum_{j> n} P(S_{j-i}=0) \sim
\begin{cases}
\frac{16}{\pi^{3/2}} n^{1/2}, &d=3, \\
\frac{8}{\pi^2} \log n, &d=4,\\
E^{\otimes 2}[L_{\infty}(S,\tilde S)], &d\ge 5,
\end{cases}
\qquad n\to\infty,
\ee
where $L_{\infty}(S,\tilde S)$ is the total intersection local time of two independent random 
walks (which is finite for $d\ge 5$).

\smallskip
The above observation \eqref{EQn} is relevant when we try to guess the behavior of $f^\wsaw(u)$ as $u\downarrow 0$.
Indeed, by the subadditivity of $\log Z_n^{\wsaw}(u)$, we may write
\begin{equation}
\label{eq:subadditive}
\lambda_d u  - f^\wsaw(u) = \sup_{m} \Big\{ \lambda_d u + \frac1m \log E[e^{- u Q_m}] \Big\}.
\end{equation}
Assuming that we can expand $\frac1m \log E[e^{- u Q_m}]$ as $u\downarrow 0$ (we will also 
take $m \asymp 1/u$), we get
\begin{equation}
\label{eq:cumuexp}
\begin{aligned}
&\log E[e^{-u Q_m}]  = \log \Big( 1- u E[Q_m] +\tfrac12 u^2 E[Q_m^2] 
- \tfrac16 u^3 E[Q_m^3]  + \ldots\Big) \\
& = -u E[Q_m]  +\tfrac12 u^2 \big(E[Q_m^2] -E[Q_m]^2  \big) 
- u^3 \big( \tfrac16 E[Q_m^3]  - \tfrac12 E[Q_m] E[Q_m^2] + \tfrac13 E[Q_m]^3  \big) + \ldots
\end{aligned}
\end{equation}

For $d=3$ we may use \eqref{EQn} and \eqref{eq:asympQalt} to get
\be
\begin{aligned}
\frac{1}{m} \log E[e^{-u Q_m}] +u \lambda_d 
& =  [1+o(1)]\, \tilde c_3 u m^{-1/2} + \tfrac12 u^2 \log m + c'_3 u^3  m^{3/2} +  c_4 u^4 m^{5/2} + \ldots\\
& = u \Big( \tilde c_3 m^{-1/2} + C u \log m +  c'_3 u^2 m^{3/2}  + c''_3 u^3 m^{5/2} + \ldots \Big)  \, .
\end{aligned}
\ee
Note that in \eqref{eq:cumuexp}, in the term of order $u^3$, the leading order is $m^3$ but the 
different terms cancel each other out: the next order is  $m^{5/2}$ because of \eqref{EQn} and 
\cite[Eq.(6.4.3)]{C10} (a similar reasoning holds for the terms of order $u^k$ with $k>3$). When 
trying to optimise over $m$, we realise that we need to take $u^2 m^{3/2} \asymp m^{-1/2}$ (and 
the term $u \log m$ will turn out to be negligible): taking $m = c u^{-1}$  (where the constant 
$c$ is chosen so as to optimise the parenthesis above), we get that $\frac{1}{m} \log E[e^{-u Q_m}] 
+u \lambda_d  \sim a_3 u^{3/2}$, which when substituted into \eqref{eq:subadditive} gives the 
conjectured behaviour.

For $d=4$, we similarly have
\be
\begin{aligned}
\frac{1}{m} \log E[e^{-u Q_m}] +u \lambda_d = u \Big( \tilde c_4 \frac{\log m}{m} 
+ C u  + c'_4 u^2 m \log m + c''_4u^3 m^{2} \log m +\ldots \Big).
\end{aligned}
\ee
To optimize over $m$, we choose $u^2 m \log m \asymp \log m /m$ (and the term $Cu$ will be 
negligible), so that taking $m = c u^{-1}$ we have $\frac{1}{m} \log E[e^{-u Q_m}] +u \lambda_d  
\sim a_4 u^2 \log 1/u $.

For $d\ge5$, we have 
\be
\begin{aligned}
\frac{1}{m} \log E[e^{-u Q_m}] +u \lambda_d = u \Big( \frac{\tilde c_d}{m}+ C u  
+ c'_d u^2 m + c''_d u^3 m^2 +\ldots \Big).
\end{aligned}
\ee
We choose $u^2 m \asymp 1/m$, so that taking $m = c u^{-1} $ (all the terms contribute) we 
have $\frac{1}{m} \log E[e^{-u Q_m}] +u \lambda_d  \sim a_d u^2$.

\section{Large deviations for the trimmed range of simple random walk}
\label{app:range}

In Section~\ref{s:toconj_sharpbetac} we explained how we would prove 
Conjecture~\ref{conj:sharp_beta_c} via Conjecture~\ref{conj:ubcritcurve}. 
In this appendix we explain how the latter follows from an estimate on 
the upper large deviations for the trimmed range, which we state as 
Conjecture~\ref{conj.LDrange} below.

\subsection{Conjecture on the upper large deviations}

It was shown by Hamama and Hesten~\cite{HK01} that the range $R_n$ of simple random 
walk satisfies an upward large deviation principle for $d \geq 2$. Namely, they showed that 
the limit
\begin{equation}
\label{Jsdef}
J(s) = \lim_{n\to\infty} \left[-\frac{1}{n} \log P(R_n \geq sn)\right], \qquad s \in [0,1],
\end{equation}   
exists, with $s \mapsto J(s)$ finite, non-negative, non-decreasing and convex on $[0,1]$, 
and (see Fig.~\ref{fig-ratefuncalt})
\begin{equation}
d=2\colon \quad J(s) > 0, \quad s>0, 
\qquad d \geq 3\colon \quad J(s) \left\{\begin{array}{ll}
= 0, &s \leq 1/\gl_d,\\
>0, &1/\lambda_d < s \leq 1.
\end{array}
\right.
\end{equation}
This is the analogue of Proposition~\ref{pr:largedevsuml2}.

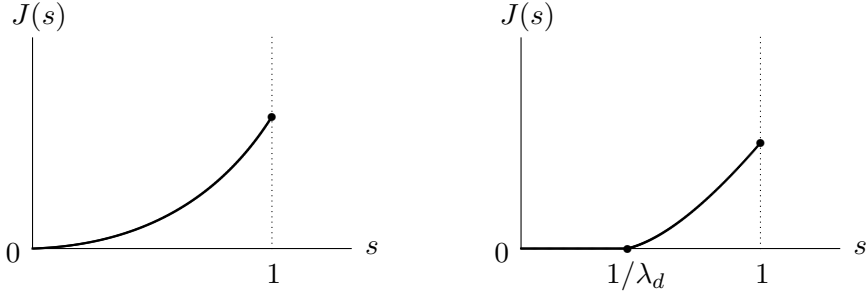
\begin{figure}[htbp]
\begin{center}
\setlength{\unitlength}{0.35cm}
\begin{picture}(12,12)(4,-2)
\put(0,0){\line(12,0){12}}
\put(0,0){\line(0,8){8}}
{\thicklines
\qbezier(0,0)(6,0.2)(9,5)
}
\qbezier[40](9,0)(9,4)(9,8)
\put(-1,-.5){$0$}
\put(12.5,-0.2){$s$}
\put(-0.8,8.5){$J(s)$}
\put(8.8,-1.3){$1$}
\put(9,5){\circle*{.35}}
\end{picture}
\begin{picture}(12,12)(-2,-2)
\put(0,0){\line(12,0){12}}
\put(0,0){\line(0,8){8}}
{\thicklines
\qbezier(0,0)(2,0)(4,0)
\qbezier(4,0)(6,0.5)(9,4)
}
\qbezier[40](9,0)(9,4)(9,8)
\put(-1,-.5){$0$}
\put(12.5,-0.2){$s$}
\put(-0.8,8.5){$J(s)$}
\put(8.8,-1.3){$1$}
\put(3.2,-1.3){$1/\gl_d$}
\put(4,0){\circle*{.35}}
\put(9,4){\circle*{.35}}
\end{picture}
\end{center}
\vspace{0cm}
\caption{\small Qualitative plots of $s \mapsto J(s)$ for $d=2$ and $d \geq 3$.}
\label{fig-ratefuncalt}
\end{figure}

Since $Q_n \leq n^2/R_n$, it follows that $J(s) \geq I(1/s)$, $s \in (0,1]$, with $I$ 
the rate function in \eqref{Itdef}. For $d=2$, $J$ inherits from $I$ the asymptotics 
found in \eqref{It2scal}, namely,
\begin{equation}
\label{Js2scal}
d = 2\colon\quad \lim_{s \downarrow 0} \big[-s \log J(s)\big] = \frac{1}{\gl_2}.
\end{equation}
Indeed, the upper bound is immediate from the corresponding upper bound on 
$-\tfrac1s \log I(1/s)$ in \eqref{It2scal}. The lower bound follows from an easy 
adaptation of the argument used in Section~\ref{ss:lem1} to prove the upper 
bound on $I(t)$. See, in particular, Step 4 in the proof of Proposition~\ref{pr:largedevsuml2}. 

\smallskip
The following conjecture deals with the upward large deviations of the range \emph{trimmed when
the local times exceed a certain threshold}. Our estimates on the rate function are 
not as good as \eqref{Jsdef}--\eqref{Js2scal}, but sufficient for our purpose. 

\begin{conjecture}
\label{conj.LDrange}
For $n\in\N$ and $A \in \N$, let
\begin{equation}
R_{n,A}^- = \{x\in\bbZ^d\colon 1\leq \ell_n(x) \leq A \}, 
\quad \gamma^-_{n,A} = \sum_{x\in R_{n,A}^-} \ell_n(x).
\end{equation}
For every $A \in \N$ and $s \in [0,1]$ there exists $J(A,s)$ such that,
\begin{equation}
P\Big( |R_{n,A}^-| \geq s \theta n,\,
\gamma_{n,A} \leq \theta n\Big) \leq e^{-J(A,s)\,\theta n},
\qquad \theta>0,\, n \geq n_0(A,s,\theta),
\end{equation}
with 
\begin{equation}
d = 2\colon \quad J(A,s)>0, \quad s>0, \qquad 
d \geq 3\colon \quad J(A,s) \left\{\begin{array}{ll}
= 0, &0 \leq s \leq 1/\gl_d(A),\\
> 0, &1/\gl_d(A) < s \leq 1,
\end{array}
\right.
\end{equation}
where
\begin{equation}
\label{eq:Ascal}
\begin{aligned}
&d=2\colon\quad \lim_{s \downarrow 0} -{s \log J(A(s),s)} 
= \frac{1}{\gl_2}, \quad A(s) \gg s^{-10},\\
&d \geq 3\colon\quad \lambda_d(A) < \lambda_d, \quad \lim_{A\to\infty} \gl_d(A) = \gl_d.
\end{aligned} 
\end{equation}
\end{conjecture}

\subsection{Towards a proof of the conjecture} 

We now propose a proof of Conjecture~\ref{conj:ubcritcurve} based on 
Conjecture~\ref{conj.LDrange}.

Recall \eqref{defZbar} and the statement of Conjecture~\ref{conj:ubcritcurve}. The 
idea is that if all the local times are small, then we get in the exponential $-f^\wsaw(u) 
+ (1-\gep) f^\wsaw(u) <0$, while if all the local times are large, then we get $0$ 
because of the indicator. We have to show that a mixture of small and large local times
contributes something in between, i.e., ``rough local-time profiles'' are costly. To that
end, decompose the range of simple random walk into two parts, corresponding to 
small and large local times:
\begin{equation}
\label{trim}
R_n^- = R_n^-(u) = \{x\in\bbZd\colon\,\ell_n(x) \leq 1/\sqrt{u}\},
\qquad R_n^+ = R_n^+(u) = \bbZ^d\backslash R_n^-(u).
\end{equation}
Using this splitting, we may write
\begin{equation}
\label{defZnugep}
\bar Z_{n,u}^{\gep} = E\Big[ e^{\sum_{x\in R_n^-} 
[-u\ell_n(x)^2 + (1-\gep)f^\wsaw(u)\ell_n(x)] } \Big].
\end{equation}
Let
\begin{equation}
\label{defgn}
\gamma_n^- = \sum_{x\in R_n^-} \ell_n(x) = \sum_{i=1}^n \ind_{\{S_i \in R_n^-\}}
\end{equation}
be the time spent in $R_n^-$. Decompose $\bar Z_{n,u}^{\gep}$ according to the value 
taken by $\gamma_n^-$:
\begin{equation}
\label{Znugepdec}
\bar Z_{n,u}^{\gep} = \sum_{k=0}^{1/\eta} 
\bar Z_{n,u}^{\gep}\Big( \frac{\gamma_n^-}{n} \in [k\eta, (k+1)\eta) \Big), 
\qquad \eta>0,\,1/\eta \in \N.
\end{equation}
We know that
\begin{equation}
\label{Znugepkzero}
\bar Z_{n,u}^{\gep} \Big( \frac{\gamma_n^-}{n} \in [0,\gd) \Big) \leq 
e^{(1-\gep)f^\wsaw(u)\,\gd n}.
\end{equation}
Suppose for now that we have the following lemma (we explain below how it follows 
from Conjecture~\ref{conj.LDrange}):

\begin{lemma}
\label{lem:rough}
For every $\gep>0$, $\eta< \gep^3, k\geq \gep^{-2}$ and $0<u \leq u_0(\gep)$, 
\begin{equation}
\label{Znutarget}
\bar Z_{n,u}^{\gep}\Big( \frac{\gamma_n^-}{n} \in [k\eta, (k+1)\eta) \Big) 
\leq e^{-\frac12 \gep(1-2\gep) f^\wsaw(u)\, k \eta\, n},
\qquad k \in \N.
\end{equation}
\hfill \qed
\end{lemma}

\noindent
Combining \eqref{Znugepdec}--\eqref{Znutarget}, we find that, splitting the sum \eqref{Znugepdec} 
at $k=1/\gep^2$,
\begin{equation}
\limsup_{n\to\infty} \frac{1}{n} \log \bar Z_{n,u}^{\gep}
\leq (1-\gep) f^\wsaw(u)\, \gep^{-2} \eta.
\end{equation}
Since $u_0(\gep)$ does not depend on $\eta$, the right-hand side tends to zero as $\eta\downarrow 0$, 
and so we get the claim in \eqref{targetpropZ}, i.e., Conjecture~\ref{conj:ubcritcurve}.

It remains to give the proof of Lemma~\ref{lem:rough} based on Conjecture~\ref{conj.LDrange} above. 

\begin{proof} 
Recall \eqref{defZnugep}--\eqref{defgn}. Estimate, abbreviating $\theta = (k+1)\eta$,
\begin{equation}
\label{Zest1}
\begin{aligned}
\bar Z_{n,u}^{\gep}\Big( \frac{\gamma_n^-}{n} \in [k\eta, (k+1)\eta) \Big) 
&= E\left[e^{\sum_{x \in R_n^-} \big[-u\ell_n(x)^2 + (1-\gep) f^\wsaw(u)\ell_n(x)\big]}
\ind_{\{\gamma_n^-\in [\theta-\eta,\theta)n\}}\right]\\
&= e^{(1-\gep) f^\wsaw(u) \theta n}\,
E\left[ e^{-uQ_n^-} \ind_{\{ \gamma_n^- \in [\theta-\eta,\theta) n \}} \right],
\end{aligned}
\end{equation}
where $Q_n^- = \sum_{x\in R_n^-} \ell_n(x)^2$. Estimate
\begin{equation}
\label{Zest2}
\begin{aligned}
E&\left[ e^{-uQ_n^-} \ind_{\{ \gamma_n^- \in [\theta-\eta,\theta) n \}} \right]\\
&= E\left[ e^{-uQ_n^-} 
\ind_{\{uQ_n^- > (1-\tfrac12\gep) f^\wsaw(u)\, \theta n\}} 
\ind_{\{ \gamma_n^- \in [\theta-\eta,\theta) n \}} \right]\\
&\qquad + E\left[ e^{-uQ_n^-}
\ind_{\{uQ_n^- \leq (1-\tfrac12\gep) f^\wsaw(u)\, \theta n\}} 
 \ind_{\{ \gamma_n^- \in [\theta-\eta,\theta) n \}} \right]\\
&\leq e^{-(1-\tfrac12\gep) f^\wsaw(u)\, \theta n}
+ P\Big(uQ_n^- \leq (1-\tfrac12\gep) f^\wsaw(u)\,\theta n,\, 
\gamma_n^- \in [\theta-\eta,\theta) n\Big).
\end{aligned}
\end{equation}
The first term in the right-hand side of \eqref{Zest2} contributes a term $e^{-\tfrac12\gep
f^\wsaw(u)\,\theta n}$ to the right-hand side of \eqref{Zest1}, which fits the estimate we 
are after. By Jensen's inequality, $Q_n^- \geq (\gamma_n^-)^2 / |R_n^-|$. Hence 
the probability in the right-hand side of \eqref{Zest2} is bounded from above by
\begin{equation}
\label{Zest3}
\begin{split}
P\left(|R_n^-| \geq \left[\frac{u}{(1-\tfrac12\gep)f^\wsaw(u)}\right]\, 
\frac{(\gamma_n^-)^2}{\theta n},\,\gamma_n^- \in [\theta-\eta,\theta) n\right) \\
\leq 
P\left(|R_n^-| \geq \left[\frac{(1-\gep^2)^2}{1-\tfrac12\gep} \frac{u}{f^\wsaw(u)}\right]
 \theta n,\,\gamma_n^- \leq \theta n\right),
\end{split}
\end{equation}
where we use that $\frac{(\gamma_n^-)^2}{\theta n} \geq (1-\tfrac\eta\theta)^2\,\theta n$ and 
choose $k\geq \gep^{-2}$ to make that $\eta/\theta = 1/(k+1) \leq \gep^2$.
 
\medskip\noindent
$\bullet$ $d \geq 3$. 
Choose $u$ small enough so that $f^\wsaw(u) \leq (1+\tfrac\gep5) \gl_d u$. 
Then, provided we fixed $\gep>0$ small enough, we have
\begin{equation}
\label{Zest4}
\eqref{Zest3} \leq P\left(|R_n^-| \geq \left[\frac{1+\tfrac14 \gep}{\gl_d}\right]
 \theta n,\,\gamma_n^- \leq \theta n\right).
\end{equation}
By Conjecture~\ref{conj.LDrange}, the latter probability is bounded from above by 
$e^{-c(\gep,u)\,\theta n}$ for some $c(\gep,u)>0$, provided that
\be
\frac{1}{\gl_d(1/\sqrt{u})} < \frac{1+\tfrac14 \gep}{\gl_d},
\ee
which holds when $1/\sqrt{u}$ exceeds a certain threshold $A=A(\gep)$. Hence, by
\eqref{eq:Ascal}, there is a $u_0=u_0(\gep)$ such that $f^{\wsaw} (u) \leq c(\gep,u)$ 
for all $0<u\leq u_0$, and we get that \eqref{Zest4} is smaller than $e^{- f^{\wsaw}(u) \theta n}$. 
This settles the claim in \eqref{Znutarget} because $k\eta \leq (k+1)\eta = \theta$.

\medskip\noindent
$\bullet$ $d=2$. 
Choose $u$ small enough so that $f^\wsaw (u) \leq (1+\tfrac\gep5) \gl_2 u \log (1/u)$. Then, 
provided we fixed $\gep$ small enough, we have
\begin{equation}
\label{Zest5}
\eqref{Zest3} \leq P\left(|R_n^-| \geq \left[\frac{1+\tfrac14 \gep}{\gl_2 \log(1/u)}\right]
 \theta n,\,\gamma_n^- \leq \theta n\right).
\end{equation}
By Conjecture~\ref{conj.LDrange}, the latter probability is bounded from above by $e^{- c(\gep, u)
\theta n}$ with $\log c(\gep,u)$ $\geq  -(1+\tfrac\gep5) \frac{\log(1/u)}{1+\gep/4}$ for $u$ sufficiently 
small. In particular, $c(\gep,u) \geq u^{1-\gep/20} \gg f^{\wsaw}(u) $ as $u\downarrow 0$. 
Consequently, there is an $u_0=u_0(\gep)$ such that \eqref{Zest3} is smaller than $e^{-f^{\wsaw}(u)
\theta n}$ for $0<u\leq u_0$. This again settles the claim in \eqref{Znutarget}.
\end{proof}


\end{appendix}



\end{document}